\documentclass{article}

\usepackage[preprint]{neurips_2026}

\usepackage[utf8]{inputenc} 
\usepackage[T1]{fontenc}    
\usepackage{hyperref}       
\usepackage{url}            
\usepackage{booktabs}       

\usepackage{amsfonts}       
\usepackage{algorithm}     
\usepackage[noend]{algpseudocode}
\usepackage{algorithm}
\usepackage{amssymb}
\usepackage{nicefrac}       
\usepackage{microtype}      
\usepackage{xcolor}         
\usepackage{amsmath} 
\usepackage{graphicx}
\usepackage{amsthm}
\usepackage{bm}
\usepackage{bbm}
\usepackage{multirow}
\usepackage{subcaption}

\newtheorem{lemma}{Lemma}
\newtheorem{proposition}{Proposition}

\title{Kinetic-Optimal Scheduling with Moment Correction \\ for Metric-Induced Discrete Flow Matching \\ in Zero-Shot Text-to-Speech}

\author{%
  Dong Yang\textsuperscript{1*}, 
  Yiyi Cai\textsuperscript{2}, 
  Haoyu Zhang\textsuperscript{1}, 
  Yuki Saito\textsuperscript{1}, 
  Hiroshi Saruwatari\textsuperscript{1} \\
  \textsuperscript{1}The University of Tokyo, \textsuperscript{2}Independent Researcher \\
  \textsuperscript{*}\texttt{ydqmkkx@gmail.com} \\
}

\begin{document}

\maketitle

\begin{abstract}
Metric-induced discrete flow matching (MI-DFM) exploits token-latent
geometry for discrete generation, but its practical use is limited by two
issues: heuristic schedulers requiring hyperparameter search, and
finite-step path-tracking error from its first-order continuous-time Markov chain (CTMC) solver. We address both issues. 
First, we derive a kinetic-optimal scheduler for prescribed scalar-parameterized probability paths, and instantiate it for MI-DFM as a training-free numerical schedule that traverses the path at constant Fisher–Rao speed. 
Second, we introduce a finite-step moment correction that adjusts the jump probability while preserving the CTMC jump destination distribution. 
We validate the resulting method, GibbsTTS, on codec-based zero-shot text-to-speech (TTS). 
Under controlled comparisons with a unified architecture and large-scale dataset, GibbsTTS achieves the best objective naturalness and is preferred in subjective evaluations over masked discrete generative baselines. 
Additionally, in comparison with the evaluated state-of-the-art TTS systems, GibbsTTS shows strong speaker similarity, achieving the highest similarity on three of four test sets and ranking second on the fourth.

Project page: https://ydqmkkx.github.io/GibbsTTSProject
\end{abstract}

\section{Introduction}
\label{sec: introduction}

Metric-induced discrete flow matching (MI-DFM)~\cite{mi-dfm} provides a principled way to use the geometry of discrete tokens in generative modeling. 
Unlike mask-source discrete generative models that corrupt data by replacing tokens with a special mask state, MI-DFM starts from a uniform token distribution and defines intermediate distributions as Gibbs distributions over token-latent distances, which gradually concentrate on the target token.
Although it has been applied to several tasks~\cite{fudoki, ursa-metric-video, wam-flow, next-omni}, its broader application remains limited by two issues. First, its time scheduler is typically heuristic and tuned through hyperparameter search. This scheduler controls how fast the Gibbs distribution concentrates toward the target token, and therefore has a strong impact on sampling quality. Second, MI-DFM inference relies on a continuous-time Markov chain (CTMC) solver derived from infinitesimal dynamics. In practice, it is applied with a finite number of inference steps, where the standard first-order solver can incur significant path-tracking error.

This work first addresses the scheduling problem from a Fisher--Rao geometric perspective. Prior work~\cite{mi-dfm} derives kinetic-optimal velocities for arbitrary discrete probability paths, and mask-source mixture paths~\cite{dfm, generalized-masked-diffusion} admit closed-form kinetic-optimal schedules. However, such closed-form schedules do not generally exist for arbitrary discrete probability paths. We therefore study kinetic-optimal time scheduling for prescribed scalar-parameterized paths \(p(x;\kappa)\). We show that the kinetic-optimal scheduler traverses the path at constant Fisher--Rao speed. For MI-DFM, the Fisher information depends on the full token-distance distribution, so we construct the scheduler numerically from distance matrices computed from token embeddings. This gives a training-free scheduler and avoids heuristic scheduler families or downstream hyperparameter search.

We further introduce a finite-step moment correction for CTMC sampling. The correction preserves the CTMC jump destination distribution and adjusts only the jump probability, so that the one-step update better matches a reference moment at the next time step. For MI-DFM, we use a local Fisher--Rao tangent statistic, yielding a lightweight correction that improves finite-step path tracking, with consistent empirical gains shown in Section~\ref{sec: results and analysis}.

We validate the resulting method on codec-based~\cite{encodec} zero-shot text-to-speech (TTS), where discrete flow matching remains less explored. Codec-based TTS is a natural application for MI-DFM, since acoustic tokens are associated with learned codebook embeddings and nearby tokens often correspond to acoustically similar representations. To evaluate the algorithmic effect, we build a DiT-based~\cite{dit} codec-token TTS model and conduct controlled comparisons under a unified architecture and large-scale dataset. Experiments on English and Chinese zero-shot TTS show that our method achieves the best objective naturalness against masked discrete generative baselines, and is further preferred in subjective evaluations.
Our contributions are summarized as follows:
\begin{itemize}
    \item 
    We formulate kinetic-optimal time scheduling for prescribed scalar-parameterized discrete probability paths, complementing prior work~\cite{mi-dfm} on kinetic-optimal velocities.
    We instantiate this framework for MI-DFM and derive a numerical scheduler from token-latent distance matrices. This avoids heuristic scheduler families and downstream hyperparameter search, which is a limitation identified in~\cite{mi-dfm}.

    \item 
    We introduce a finite-step moment correction for CTMC sampling that reduces path-tracking error by adjusting the jump probability while preserving the CTMC jump destination distribution.

    \item 
    We build a DiT-based codec-token TTS model and conduct controlled comparisons under a unified architecture and large-scale dataset, providing a fair evaluation of MI-DFM against masked discrete generative baselines. To the best of our knowledge, this is the first study of MI-DFM for TTS.

    \item 
    Experiments validate the proposed methods. Comparisons against state-of-the-art zero-shot TTS systems further indicate the potential of MI-DFM for TTS and support the proposed model as a credible testbed for evaluating discrete generative algorithms.
\end{itemize}

For clarity in the following sections and to facilitate future comparisons, we refer to the full model equipped with the kinetic-optimal scheduler and finite-step moment correction as \textit{GibbsTTS}, reflecting the Gibbs-distribution form of the metric-induced path. 
\section{Preliminaries}
\label{sec: preliminaries}

\subsection{Discrete flow matching}
\label{sec: dfm}

\textbf{Probability paths.}
We denote a discrete token sequence $x = (x^1, x^2, ..., x^N) \in \mathcal{S} = [s]^N$, $[s] = \{1, ..., s\}$, where $N$ is the length of the sequence, $s$ is the vocabulary size of the tokens, and $\mathcal{S}$ is the set of possible sequences. Let $p(x)$ and $q(x)$ denote the source and target \textit{probability mass functions} (PMFs) over the space $\mathcal{S}$, the goal of discrete flow matching is to learn the transformation from $p(x)$ to $q(x)$. We consider time-dependent \textit{probability paths} $p_t(x)$, $t \in [0, 1]$ with the form
\begin{equation}
    p_t(x) = \sum_{x_1 \in \mathcal{S}} p_t(x | x_1) q(x_1), 
    \quad
    p_t(x | x_1) = \prod_{i=1}^N p_t(x^i | x_1^i)
\end{equation}
where $p_t(x^i | x_1^i)$ is a conditional probability path, and the boundary conditions are $p_0(x^i | x_1^i) = p_0(x^i)$ and $p_1(x^i|x_1^i) = \delta_{x_1^i}(x^i)$. 
Here, $\delta_y$ denotes the delta function on $\mathcal{S}$, defined by
\begin{equation}
    \delta_y(x) = \textstyle\prod_{i=1}^N \delta_{y^i}(x^i), \
    \text{where} \
    \delta_{y^i}(x^i) = 
    \begin{cases}
        1 & x^i = y^i \\
        0 & x^i \neq y^i
    \end{cases}.
\end{equation}
The mixture paths~\cite{dfm} with $x_1^i$-dependent schedulers introduced in~\cite{generalized-masked-diffusion} are defined as
\begin{equation}
    p_t(x^i | x_1^i) = (1 - \kappa_t(x_1^i)) p(x^i) + \kappa_t(x_1^i) \delta_{x_1^i}(x^i), 
    \quad
    \kappa_0(\cdot) = 0
    \text{ and }
    \kappa_1(\cdot) = 1.
    \label{eq: mixture pt}
\end{equation}
\textbf{Probability velocities.} 
We consider a continuous-time Markov chain (CTMC) $X_t \sim p_t$, $t \in [0,1)$ in $\mathcal{S}$, and each token in $X_t$ is updated independently with the time step $h > 0$:
\begin{equation}
\mathbb{P}(X_{t+h} = x | X_t = z) 
    = \delta_{z}(x) + h u_t(x, z) + o(h),
\
u_t(x, z) 
    = \textstyle\sum_{i=1}^N u_t^i(x^i, z) \textstyle\prod_{j \neq i}\delta_{z^j}(x^j),
\end{equation}
where $o(h)$ denotes a higher-order infinitesimal and is omitted in the following derivation. $u_t$ is called the \textit{probability velocity}. When $x^i \neq z^i$, $u_t^i(x^i, z) \geq 0$ and can be expressed as a marginal form:
\begin{equation}
u_t^i(x^i, z) 
    = \tfrac{\dot{\kappa}_t}{1 - \kappa_t} [p_{1 | t}^i(x^i | z) - \delta_{z^i}(x^i)],
\quad
p_{1 | t}^i(x^i | z) 
    = \textstyle\sum_{x_1 \in \mathcal{S}}\delta_{x_1^i}(x^i) \tfrac{p_t(z | x_1) q(x_1)}{p_t(z)}, \label{eq: p_1_t}
\end{equation}
where $p_{1 | t}(x^i | z)$ is the posterior probability.
That is, $u_t^i(x^i, z) = \mathbb{E}_{x_1 \sim p_{1|t}(\cdot | z)} [u_t^i(x^i, z^i | x_1^i)]$.

According to the rate condition, when $x^i = z^i$,
\begin{equation}
    u_t^i(z^i,z^i|x_1^i) = -\sum_{x^i \neq z^i} u_t^i(x^i,z^i|x_1^i).
    \label{eq: rate condition}
\end{equation}
\textbf{Training.} 
A simple training objective of DFM is
\begin{equation}
    \mathcal{L}(\theta)
    = \mathbb{E}_{t \sim \mathcal{U}[0,1], x_1 \sim q(\cdot), x \sim p_t(\cdot | x_1)}
    \textstyle\sum_{i=1}^N [-\log p_{1 | t}^{\theta,i} (x_1^i | x)].
\end{equation}
The learnable parameters $\theta$ are trained to predict the target states from the states sampled along the probability paths.

\textbf{Inference.} At inference, \cite{mi-dfm} proposes the following sampling process that avoids the summation computation in Eq.~\ref{eq: p_1_t}:
\begin{equation}
    \quad
    \hat x_1^i \sim p_{1 | t}^{\theta,i}(\cdot | x_t), 
    \quad
    \hat x_{t+h}^i \sim \delta_{x_t^i}(\cdot) + h u_t^i(\cdot, x_t^i | \hat x_1^i)
\end{equation}

\subsection{Kinetic-optimal velocities and probability paths}
\label{sec: KO u&k}

\textbf{Notation.} 
From here on, we use a simplified notation following~\cite{mi-dfm}, where the state space is now $[s]$ and $x, z \in [s]$. The position superscript $i$ is omitted for clarity.

\textbf{Kinetic-optimal $u_t^*$.}
\cite{mi-dfm} proposes a general formulation to construct a kinetic-optimal velocity $u_t^*(x)$ for an arbitrary discrete probability path $p_t(x)$, from the perspective of minimizing the symmetric kinetic energy:
\begin{equation}
    u_t^*(x,z)
    =
    \begin{cases}
    \dfrac{
    \left[
    p_t(z)\dot p_t(x)-\dot p_t(z)p_t(x)
    \right]_+
    }{
    p_t(z)
    },
    & \text{if } x\neq z \text{ and } p_t(z)>0,\\[1.0em]
    0,
    & \text{if } x\neq z \text{ and } p_t(z)=0.
    \end{cases}
    \label{eq: ko velocity}
\end{equation}
The diagonal where $x=z$ is then determined by the rate condition in
Eq.~\ref{eq: rate condition}.

\textbf{Kinetic optimal $p_t^*$.} 
\cite{mi-dfm} formulates the problem of finding the kinetic-optimal
probability path by introducing $a_t(x)=\sqrt{p_t(x)}$:
\begin{equation}
    \min_{a} \int_{0}^{1}\sum_{x}\dot{a}_{t}(x)^{2}dt,
    \quad
    \text{s.t.} \sum_{x}a_{t}(x)^{2}=1, 
    \text{ }
    \forall t\in[0,1],
    \text{ }
    a_{0}(x)=\sqrt{p(x)}, 
    \text{ }
    a_{1}(x)=\sqrt{q(x)},
\label{eq: minimize ko}
\end{equation}
The solution is the geodesic on the probability hypersphere, and the
optimal probability path is $p_t^*(x)=a_t(x)^2$.
When $q(x)=\delta_{x_1}(x)$, this path takes the mixture form in
Eq.~\ref{eq: mixture pt}, where
\begin{equation}
    \kappa_{t}(x_{1})=1-\frac{\sin^{2}((1-t)\Omega(x_{1}))}{\sin^{2}\Omega(x_{1})}, 
    \quad \text{where } \Omega(x_{1})=\arccos\sqrt{p(x_{1})}.
\label{eq: mixture ko}
\end{equation}
When it is used with a mask source distribution that $p(x) = \delta_\text{mask}(x)$, then $p(x_1) = 0$ and $\Omega(x_1) = \frac{\pi}{2}$, thus $\kappa_{t}(x_{1}) = \sin^2(\frac{\pi}{2}t)$.


\subsection{Metric-induced discrete flow matching}
\label{sec: mi-dfm}

\cite{mi-dfm} further introduces a metric-induced $p_t(x)$ in Eq~\ref{eq: metric-induced pt}, which we call metric-induced DFM (MI-DFM). 
It exploits the geometric structure of the tokenizer latent space with the form:
\begin{align}
    p_t(x | x_1) &= \mathrm{softmax}(-\beta_t d(x, x_1)),
    \label{eq: metric-induced pt}
\end{align}
where $\beta: [0,1] \to \mathbb{R}_{\geq 0}$ is a monotonic scheduler with $\beta_0 = 0$, $\beta_1 = \infty$, and $d: [s] \times [s] \to \mathbb{R}_{\geq 0}$ such that $d(x, x_1) = 0 \Leftrightarrow x = x_1$.
This $p_t$ takes the form of a Gibbs distribution,
where the energy of each token is determined by its distance to the target token $x_1$ in the tokenizer latent space.
At $t=0$, all tokens have equal sampling probability, corresponding to random token initialization,
whereas as $t \to 1$, the distribution collapses to $\delta_{x_1}(x)$.
Combining Eq.~\ref{eq: ko velocity} with Eq.~\ref{eq: metric-induced pt},
the $u_t^*$ is given by
\begin{equation}
    u_t^*(x, z | x_1) = p_t(x | x_1)\dot{\beta}_t [d(z, x_1) - d(x, x_1)]_+.
    \label{eq: optimal ut}
\end{equation}
The distance function $d(\cdot,\cdot)$ can be instantiated using any valid metric.
A direct choice is the $\ell_p$-norm distance between the latent embeddings $\mathbf{e}(\cdot)$ associated with tokens, where $\beta_t$ is in a heuristic form, and $\text{lp}, c, a$ are hyperparameters:
\begin{equation}
    d(x, x_1) = | \mathbf{e}(x) - \mathbf{e}(x_1)|_2^{\text{lp}},
    \quad 
    {\beta}_t = c (\frac{t}{1 - t})^a.
\label{eq: distance and beta}
\end{equation}



\section{Kinetic-optimal time scheduling}
\label{sec: method}

In this section, we study the kinetic-optimal time scheduling problem for a prescribed probability path.
We consider a fixed one-dimensional family of categorical distributions
$p(x;\kappa)$, where $\kappa \in [0,\kappa_{\max}]$ is a scalar path parameter
that specifies a geometric curve on the probability simplex.
A monotonic time scheduler $\kappa_t$ induces the time-dependent probability path
$
    p_t(x) = p(x;\kappa_t),
\text{ }
    t \in [0,1].
$
Here, we do not optimize over all possible probability paths. Instead, the geometric curve
$p(x;\kappa)$ is fixed, and our goal is to find the kinetic optimal
$\kappa_t^*$ that minimizes the Fisher--Rao path energy:
$
    \kappa_t^*
    =
    \arg\min_{\kappa_t \in \mathcal{K}}
    \mathcal{E}_{FR}[\kappa_t],
    \text{ }
    \mathcal{K}
    =
    \left\{
    \kappa_t:
    \kappa_0=0,\ 
    \kappa_1=\kappa_{\max},\
    \dot{\kappa}_t \geq 0
    \right\}.
$
Under this formulation, changing $\kappa_t$ only reparameterizes the same geometric path.
Therefore, the Fisher-Rao length is invariant to the scheduler, while the Fisher-Rao energy depends on the speed along the path. The kinetic-optimal scheduler is obtained by enforcing a constant Fisher--Rao speed.

To measure the cost of the probability path in the discrete state space, we utilize the path energy induced by the Fisher-Rao metric~\cite{fisher-amari}. For a categorical distribution, the squared Riemannian velocity of the probability path $p_t$ is given by $\|\dot{p}_t\|^2_{g_{FR}}=\sum_x\frac{\dot{p}_t(x)^2}{p_t(x)}$. Therefore, the Fisher-Rao path energy over $t \in [0,1]$ can be explicitly formulated as
$
    \mathcal{E}_{FR}[p]
    = \int_0^1 \|\dot{p}_t\|^2_{g_{FR}}dt
    = \int_0^1 \sum_x\frac{\dot{p}_t(x)^2}{p_t(x)}dt
$

This energy is consistent with the kinetic energy defined in Eq.~\ref{eq: minimize ko}. Specifically, with
$a_t(x)=\sqrt{p_t(x)}$, we have
$
    \dot{a}_t(x)
    =
    \frac{\dot{p}_t(x)}{2\sqrt{p_t(x)}},
    \text{}
    \mathcal{E}_{FR}[p]
    =
    4\int_0^1 \sum_x \dot{a}_t(x)^2 dt.
$
Therefore, the kinetic energy minimization problem defined in Eq.~\ref{eq: minimize ko} is equivalent to minimizing the Fisher-Rao path energy.
In our setting, the geometric path $p(x;\kappa)$ is prescribed, and we only optimize its time reparameterization $\kappa_t$.

\subsection{Kinetic-optimal scheduling via Fisher-Rao geometry}

We derive the kinetic-optimal scheduler for a prescribed geometric path
$p(x;\kappa)$. 
The following lemmas and proposition show that the solution is obtained by reparameterizing the path with constant Fisher--Rao speed, which
leads to a constructive scheduler determined by the Fisher information.

The proofs of Lemma~\ref{lemma: length-invariance}, Lemma~\ref{lemma: constant-speed}, and Proposition~\ref{proposition: ko scheduler} are provided in Appendix~\ref{appendix: proof}.

\begin{lemma}[Scheduler-invariant Fisher-Rao length]
\label{lemma: length-invariance}
Let $p_t(x)=p(x;\kappa_t)$, where $p(x;\kappa)$ is a differentiable probability path with $\kappa\in[0,\kappa_{\max}]$, and $\kappa_t$ is a monotonic scheduler with $\kappa_0=0$ and $\kappa_1=\kappa_{\max}$. 
Then the Fisher-Rao length is invariant to the choice of $\kappa_t$:
\begin{equation}
    \ell[\kappa_t]
    =
    \int_0^1 \|\dot p_t\|_{g_{FR}}dt
    =
    \int_0^{\kappa_{\max}}\sqrt{\mathcal I(\kappa)}d\kappa
    \triangleq L,
\end{equation}
where
$
    \mathcal{I}(\kappa)
    =
    \sum_x
    \frac{
    \left(\partial_\kappa p(x;\kappa)\right)^2
    }{
    p(x;\kappa)
    }
$
is the Fisher information of the path parameter $\kappa$.
\end{lemma}

\begin{lemma}[Constant-speed characterization]
\label{lemma: constant-speed}
For a prescribed probability path $p(x;\kappa)$, let
$p_t^*(x)=p(x;\kappa_t^*)$ be the path induced by the kinetic-optimal scheduler $\kappa_t^*$. Then $\kappa_t^*$ minimizes $\mathcal{E}_{FR}[\kappa_t]$ over $\kappa_t\in\mathcal{K}$ if and only if it satisfies
\begin{equation}
    \|\dot p_t^*\|_{g_{FR}} = L,
    \qquad \forall t\in[0,1],
\end{equation}
where $L$ is the scheduler-invariant Fisher-Rao length defined in
Lemma~\ref{lemma: length-invariance}.
\end{lemma}

\begin{proposition}[Kinetic-optimal time scheduler]
\label{proposition: ko scheduler}
Let $p_t(x)=p(x;\kappa_t)$ be a prescribed path parameterized by a monotonic scheduler $\kappa_t$, with $\kappa_0=0$ and $\kappa_1=\kappa_{\max}$.
Assume $\mathcal{I}(\kappa)>0$ on $[0,\kappa_{\max}]$. The kinetic-optimal scheduler and its time derivative are given by:
\vspace{-5pt}
\[
    \kappa_t^*
    =
    F^{-1}(t),
    \quad 
    \dot{\kappa}_t^*
    =
    \frac{L}{
    \sqrt{\mathcal{I}(\kappa_t^*)}
    },
    \quad
    \text{where }
    F(\kappa)
    =
    \frac{
    \int_0^\kappa \sqrt{\mathcal{I}(\xi)} d\xi
    }{L}.
\]
\vspace{-10pt}
\end{proposition}
The mixture path in Eq.~\ref{eq: mixture pt} already admits a closed-form kinetic-optimal scheduler (Section~\ref{sec: KO u&k};~\cite{mi-dfm}), and 
its mask-source form has also been shown to be Fisher--Rao optimal for masked discrete diffusion models~\cite{scheduler-fisher}. 
Since this case is already well understood, we focus the main text on metric-induced paths and provide the mixture-path derivation in Appendix~\ref{appendix: mixture KO} as a consistency check.

\subsection{Numerical construction for metric-induced paths}
We now instantiate Proposition~\ref{proposition: ko scheduler} for the
metric-induced path defined in Eq.~\ref{eq: metric-induced pt}.
We consider it as
$p(x \mid x_1; \beta) = \mathrm{softmax}(-\beta d(x, x_1))$, where the scheduler given by $\kappa_t = \beta_t$.

For a fixed $x_1$, the Fisher information as defined in Lemma~\ref{lemma: length-invariance} with respect to $\beta$ is
$
    \mathcal I_{x_1}(\beta)
    =
    \mathbb E_{x\sim p(\cdot\mid x_1;\beta)}
    \left[
    \left(
    \partial_\beta
    \log p(x\mid x_1;\beta)
    \right)^2
    \right].
$
Then, we obtain
\begin{equation}
\partial_\beta \log p(x\mid x_1;\beta)
    =
    \mathbb E_{y\sim p(\cdot\mid x_1;\beta)}
    \left[d(y,x_1)\right]
    -
    d(x,x_1),
    \;
    \mathcal I_{x_1}(\beta)
    =
    \mathrm{Var}_{x\sim p(\cdot\mid x_1;\beta)}
    \left[
    d(x,x_1)
    \right].
\label{eq: midfm fisher}
\end{equation}
Thus, for the metric-induced $p_t(x)$ in a Gibbs distribution form, the Fisher information $\mathcal I_{x_1}(\beta)$ is exactly the variance of the distance to the target token.
Although $\mathcal I_{x_1}(\beta)$ admits a simple variance form, it still depends on the full set of distances $\{d(x,x_1)\}_{x\in[s]}$ and the induced distribution $p(x\mid x_1;\beta)$, and therefore does not generally admit a simple analytic expression.

We therefore construct the scheduler numerically.
Since the limiting endpoint corresponds to $\beta\to\infty$, we first choose
a sufficiently large $\beta_{\max}$ to approximate this limit, as detailed in
Algorithm~\ref{algo: find_beta_max}.

We then construct a one-dimensional inverse-temperature grid
$\{\beta_i\}_{i=1}^{I}$, where
$0=\beta_1<\cdots<\beta_I=\beta_{\max}$, and evaluate the Fisher information
on this grid. The resulting values are used to compute the cumulative
Fisher--Rao arc length and invert it on a uniform time grid
$\{t_j\}_{j=1}^{T}$.
This yields the lookup tables
$\{\beta_j^*\}_{j=1}^{T}$ and
$\{\dot{\beta}_j^*\}_{j=1}^{T}$, as detailed in
Algorithm~\ref{algo: KO_schedule}. During training and inference,
$\beta_t^*$ and $\dot{\beta}_t^*$ for arbitrary $t\in[0,1]$ are obtained by
linear interpolation from these tables.
This method does not require model training to search for scheduler hyperparameters.
The remaining choices are numerical construction parameters, where $\beta_{\max}$ is automatically determined by Algorithm~\ref{algo: find_beta_max}, while the others control the numerical resolution.
Increasing the grid resolution improves the approximation accuracy but increases precomputation and storage costs, so they are selected based on an accuracy-cost trade-off rather than downstream experimental search. 

Please see more details in Appendix~\ref{appendix: algo-numer}.

\section{Finite-step moment correction}
\label{sec: method2}

The kinetic-optimal velocity in Eq.~\ref{eq: ko velocity} and its metric-induced form in Eq.~\ref{eq: optimal ut} define infinitesimal CTMC dynamics. In practice, inference uses a finite number of steps, and the standard first-order solver approximates the evolution over \([t,t+h]\) using the instantaneous rate at time \(t\). The resulting one-step transition is exact only in the infinitesimal-step limit, and may incur non-negligible path-tracking error when the number of inference steps is limited.

We introduce a finite-step moment-matched correction to reduce this practical discretization error.
The correction is not intended as a higher-order CTMC integrator or an exact finite-step transition kernel.
Instead, it keeps the normalized instantaneous CTMC jump destination distribution unchanged and adjusts only the probability of accepting a jump, so that a chosen scalar moment better matches a reference value at time \(t+h\).
We first describe the correction as a generic framework parameterized by the choice of moment statistic and reference moment, and then instantiate it for metric-induced DFM. 
As a consistency check, Appendix~\ref{appendix: mixture-consistency} shows that the generic correction reduces to the known exact finite-step update for the standard mixture path when using the target-state indicator moment and the corresponding state-conditional reference moment. 

\subsection{A generic moment-corrected CTMC step}

We use the single-token notation from Section~\ref{sec: preliminaries}.
At inference, given the current token $z$ at time $t$, we first sample a
predicted endpoint token
$\hat{x}_1 \sim p_{1|t}^{\theta}(\cdot \mid z)$
from the learned posterior. For the conditional rate
$u_t(x,z\mid \hat{x}_1)$, we define the jump intensity and the jump destination
distribution as
\begin{equation}
    \lambda_t(z\mid \hat{x}_1)
    =
    \sum_{x\neq z} u_t(x,z\mid \hat{x}_1),
    \qquad
    \pi_t(x\mid z,\hat{x}_1)
    =
    \frac{
        u_t(x,z\mid \hat{x}_1)
    }{
        \lambda_t(z\mid \hat{x}_1)
    },
    \text{ for } x\neq z .
\end{equation}
When $\lambda_t(z\mid \hat{x}_1)=0$, no jump is performed.
The first-order CTMC solver introduced in~\cite{mi-dfm} uses
$
    \rho_{\mathrm{base}}
    =
    1-\exp\!\left(-h\lambda_t(z\mid \hat{x}_1)\right)
$
as the probability of making a jump, and samples the destination from
$\pi_t(\cdot\mid z,\hat{x}_1)$.
This solver is exact only in the infinitesimal-step limit. For a finite
step size $h$, the one-step transition may deviate from the reference
distribution $p_{t+h}(\cdot\mid \hat{x}_1)$. We therefore introduce a lightweight
finite-step correction that keeps the jump destination distribution $\pi_t$
fixed, and only adjusts the jump probability.
Specifically, let $\phi_t(x\mid \hat{x}_1)$ be a scalar statistic that measures
the progress of state $x$ along the path, and let
$m_{t+h}(z,\hat{x}_1)$ be the reference moment that the one-step transition
should match. If a jump probability $\rho\in[0,1]$ is used, the post-step moment
is
\begin{equation}
    (1-\rho)\phi_t(z\mid \hat{x}_1)
    +
    \rho
    \bar{\phi}_t(z,\hat{x}_1),
    \quad
    \text{where }
    \bar{\phi}_t(z,\hat{x}_1)
    =
    \mathbb{E}_{y\sim \pi_t(\cdot\mid z,\hat{x}_1)}
    \left[
        \phi_t(y\mid \hat{x}_1)
    \right].
\label{eq: fisher phi}
\end{equation}
We first obtain the moment-matching jump probability by solving the following
one-dimensional least-squares problem:
\begin{equation}
    \rho^\star
    =
    \arg\min_{\rho\in\mathbb{R}}
    \left[
    (1-\rho)\phi_t(z\mid \hat{x}_1)
    +
    \rho
    \bar{\phi}_t(z,\hat{x}_1)
    -
    m_{t+h}(z,\hat{x}_1)
    \right]^2 .
    \label{eq: rho opt}
\end{equation}
When
$\phi_t(z\mid \hat{x}_1)\neq \bar{\phi}_t(z,\hat{x}_1)$,
the solution to Eq.~\eqref{eq: rho opt} is
\begin{equation}
    \rho^\star
    =
    \frac{
    \phi_t(z\mid \hat{x}_1)-m_{t+h}(z,\hat{x}_1)
    }{
    \phi_t(z\mid \hat{x}_1)-\bar{\phi}_t(z,\hat{x}_1)
    } .
\label{eq: generic corr rho}
\end{equation}
Since $\rho^\star$ serves as a jump probability, it must lie in $[0,1]$.
We use the moment-matching correction only when $0\leq \rho^\star\leq 1$, otherwise fall back to the original first-order CTMC solver:
\begin{equation}
    \rho_{\mathrm{corr}}
    =
    \begin{cases}
    \rho^\star,
    & \text{if } \phi_t(z\mid \hat{x}_1)\neq \bar{\phi}_t(z,\hat{x}_1) \text{ and } 0\leq \rho^\star \leq 1,\\
    \rho_{\mathrm{base}},
    & \text{otherwise}.
    \end{cases}
\label{eq: constraint corr rho}
\end{equation}
When
$\phi_t(z\mid \hat{x}_1)=\bar{\phi}_t(z,\hat{x}_1)$,
the post-step moment is independent of $\rho$. In this degenerate case, we also
use $\rho_{\mathrm{base}}$.
The statistic $\phi_t$ and the reference moment $m_{t+h}$ should be chosen
according to the structure of each probability path.


\subsection{Application to metric-induced paths}
\label{sec: midfm correction}

For the metric-induced path in Eq.~\ref{eq: metric-induced pt}, given a predicted $\hat{x}_1$, we have
$p_t(x\mid \hat{x}_1)=\mathrm{softmax}(-\beta_t d(x,\hat{x}_1))$. 
Since the purpose of the correction is to reduce the finite-step deviation from the reference path, we use the local Fisher-Rao tangent statistic:
\begin{equation}
    \phi_t(x\mid \hat{x}_1)
    =
    \partial_t\log p_t(x\mid \hat{x}_1)
    = 
    \dot{\beta}_t
    \left(
    \mathbb{E}_{y\sim p_t(\cdot\mid \hat{x}_1)}
    [d(y,\hat{x}_1)]
    -
    d(x,\hat{x}_1)
    \right).
\label{eq: fisher tangent}
\end{equation}

The corresponding reference moment is taken under the next-time distribution:
\begin{equation}
    m_{t+h}(z,\hat{x}_1)
    =
    \mathbb{E}_{y\sim p_{t+h}(\cdot\mid \hat{x}_1)}
    \left[
        \phi_t(y\mid \hat{x}_1)
    \right]
    =
    \dot{\beta}_t
    \left(
    \mathbb{E}_{y\sim p_t(\cdot\mid \hat{x}_1)}
    [d(y,\hat{x}_1)]
    -
    \mathbb{E}_{y\sim p_{t+h}(\cdot\mid \hat{x}_1)}
    [d(y,\hat{x}_1)]
    \right).
    \label{eq: fisher moment}
\end{equation}
Since
$
    \mathbb{E}_{y\sim p_t(\cdot\mid \hat{x}_1)}
    \left[
        \partial_t\log p_t(y\mid \hat{x}_1)
    \right]
    =
    0,
$
this reference moment measures the displacement from $p_t$ to $p_{t+h}$ after projection onto the local Fisher--Rao tangent direction. Note that $\phi_t$ is the tangent statistic at time $t$, so evaluating it under $p_{t+h}$ is itself a first-order approximation in $h$. Matching this moment thus aligns the one-step update with the reference path along this tangent direction, while preserving the CTMC jump destination distribution.
Then, according to Eq.~\ref{eq: fisher phi}, we have
\begin{equation}
    \bar{\phi}_t(z,\hat{x}_1)
    =
    \mathbb{E}_{y\sim \pi_t(\cdot\mid z,\hat{x}_1)}
    \left[
        \phi_t(y\mid \hat{x}_1)
    \right]
    = 
    \dot{\beta}_t
    \left(
    \mathbb{E}_{y\sim p_t(\cdot\mid \hat{x}_1)}
    [d(y,\hat{x}_1)]
    -
    \mathbb{E}_{y\sim \pi_t(\cdot\mid z,\hat{x}_1)}
    [d(y,\hat{x}_1)]
    \right).
\label{eq: fisher post moment}
\end{equation}

Substituting Eqs.~\ref{eq: fisher tangent}, \ref{eq: fisher moment}, and \ref{eq: fisher post moment} into Eq.~\ref{eq: generic corr rho} yields
\begin{equation}
    \rho^\star
    =
    \frac{
    d(z,\hat{x}_1)
    -
    \mathbb{E}_{y\sim p_{t+h}(\cdot\mid \hat{x}_1)}
    [d(y,\hat{x}_1)]
    }{
    d(z,\hat{x}_1)
    -
    \mathbb{E}_{y\sim \pi_t(\cdot\mid z,\hat{x}_1)}
    [d(y,\hat{x}_1)]
    }.
\label{eq: midfm corr rho raw}
\end{equation}
The final jump probability is obtained from Eq.~\ref{eq: constraint corr rho}.
Since \(\dot\beta_t\) cancels, the corrected jump probability depends only on expected distances and does not recover the infinitesimal CTMC limit as \(h\to0\). 
Therefore, the correction should be interpreted as a finite-step approximation aimed at reducing discretization error at practical NFEs.

\section{Model}
\label{sec: model}

\begin{figure}
  \centering
  \includegraphics[width=0.9\linewidth, clip]{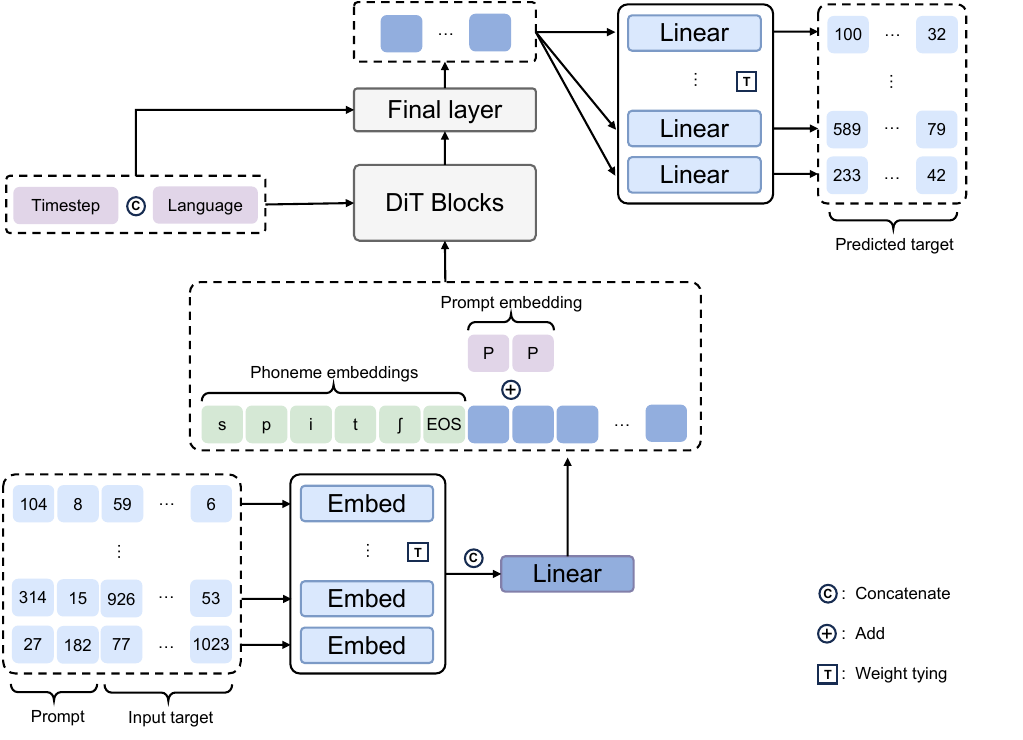}
  \caption{Architecture of the proposed model.}
  \label{fig: architecture}
\end{figure}

For conciseness, this section summarizes the main model design, while detailed training and inference procedures are provided in Appendix~\ref{appendix: algo-model}. Specifically, Algorithms~\ref{algo: train} and~\ref{algo: infer} describe the training and inference processes, respectively.

\textbf{Backbone.}
As shown in Fig.~\ref{fig: architecture}, we adopt a DiT~\cite{dit} backbone and leverage RoPE position embedding~\cite{rope}, SwiGLU~\cite{swiglu}, and RMSNorm~\cite{rmsnorm}. 
The timestep and language embeddings are concatenated and used as conditioning for the adaLN-Zero layers in DiT.

\textbf{Input construction.}
We improve the StableTTS text frontend~\cite{stabletts} for text normalization and grapheme-to-phoneme conversion. Codec-token embeddings from all RVQ codebooks are concatenated along the channel dimension and linearly projected to a per-frame embedding, which is then concatenated along the time dimension with the phoneme embeddings.

\textbf{Prompt construction.} During training, each utterance is randomly split at a ratio sampled from \(U(0,0.3)\), with the prefix used as the speech prompt and the remainder as the prediction target. A learnable prompt embedding is added to the prompt positions.

\textbf{Duration predictor.} 
We adopt a rule-based duration predictor inspired by the released implementation of MaskGCT and further refine it for our model. Details are provided in Appendix~\ref{appendix: duration predictor}.

\textbf{Codebook-wise loss weighting.}
We perform full-codebook training and inference, and introduce codebook-wise loss weights in the training objective, as shown in Lines 12--13 of Algorithm~\ref{algo: train}. We discuss its effect and compare it with the per-layer training strategy of MaskGCT in Appendix~\ref{appendix: training objective}.
\section{Experimental setup}
\label{sec: experimental setup}

\subsection{Model implementations}

\begin{table}[t]
  \caption{Model and training configurations.}
  \label{tab: train configs}
  \centering
  \vspace{2pt}
  \resizebox{1.\textwidth}{!}{
  \begin{tabular}{lcrcccccc}
    \toprule
    \textbf{Model} 
    & \textbf{Size} 
    & \textbf{Dim.} 
    & \textbf{Layers} 
    & \textbf{Dataset} 
    & \textbf{Max. length} 
    & \textbf{Max. batch hours} 
    & \textbf{GPUs} 
    & \textbf{Train time} \\
    \midrule
    Base  
    & 178M 
    & 768   
    & 12 
    & Emilia-en     
    & 1,024 tokens / 20.48 s
    & 2.912 
    & 8  
    & 33 hours \\
    Large 
    & 399M 
    & 1,024 
    & 16 
    & Emilia-en/zh  
    & 1,536 tokens / 30.72 s
    & 4.369 
    & 32 
    & 46 hours \\
    \bottomrule
  \end{tabular}
  }
\end{table}

\textbf{Variants.} As shown in Table~\ref{tab: train configs}, we set both Base and Large variants. The Base variants are used to reduce computational cost and are employed for hyperparameter search and ablation studies.

\textbf{Codec.} We use the pre-trained acoustic codec released by MaskGCT. The encoder follows DAC~\cite{dac} and the decoder follows Vocos~\cite{vocos}. It has 12 RVQ layers, a codebook size of 1024, and an embedding dimension of 8. 
With a 24,000 Hz sampling rate and a hop size of 480, 50 codec frames correspond to 1 second of speech. 

\textbf{Codebook distance matrices.}
Since the codec uses \(\ell_2\)-normalized codebook embeddings, we compute the codebook distance as the squared Euclidean distance between these normalized embeddings.
The distance matrices in Appendices~\ref{appendix: algo-numer} and~\ref{appendix: algo-model} are instantiated as \(\{\mathbf{D}_c\}_{c=1}^{C}\), where
\(\mathbf{D}_c\) contains the pairwise squared Euclidean distances between
the embeddings of the \(c\)-th RVQ codebook.

\textbf{Masked discrete generative baselines.}
As discussed in Section~\ref{sec: dfm}, masked DFM and masked discrete diffusion (DD) share the same training objective.
Therefore, for masked models, we evaluate both the standard DFM~\cite{dfm,diflow} and the DD-style inference procedure used in MaskGCT.
As shown in Table~\ref{tab: objective results}, we also compare different time schedulers for masked models: the closed-form kinetic-optimal (KO) scheduler $\kappa_t=\sin^2(\frac{\pi}{2}t)$ discussed in Section~\ref{sec: KO u&k}, the scheduler $\kappa_t=t^2$~\cite{diflow} used in DiFlow-TTS~\cite{diflow}, and the scheduler $\kappa_t=\sin(\frac{\pi}{2}t)$ used in MaskGCT~\cite{maskgct}.

\subsection{Training and inference}

We train all models for 10 epochs using AdamW~\cite{adamw} without hyperparameter tuning.
The peak learning rate is set to $2\times10^{-4}$, with linear warmup over the first $5\%$ of training steps, followed by cosine decay to $10\%$ of the peak learning rate.
We maintain an exponential moving average (EMA) of the model weights with a decay rate of $0.9999$.
Following MaskGCT, we apply classifier-free guidance (CFG) with a condition drop rate of $0.15$, a CFG scale of $2.5$, and a rescale factor of $0.75$.
During training, the DiT backbone uses BF16 precision, except for the transformation in RMSNorm layers, which is performed in FP32. All other computations are performed in FP32.
All training, inference, and objective evaluations are conducted on NVIDIA H100 GPUs with 96 GB memory.

\textbf{Dynamic batching.}
We adopt a dynamic batching strategy. In Table~\ref{tab: train configs}, the maximum length denotes the maximum sequence length of acoustic tokens.
Given this maximum length, the number of GPUs, and a gradient accumulation factor of $2$ for the Base variants, the maximum amount of audio per batch is $2.912$ hours for the Base variants and $4.369$ hours for the Large variants.

\textbf{Training dataset.} We use the English (en) and Chinese (zh) subsets of Emilia~\cite{emilia} for training. After preprocessing, their total durations are over 46k and 45k hours, respectively.

\textbf{Inference.}
We use 32 function evaluations (NFEs) in the main experiments, and report results with 16 and 64 NFEs in Appendix~\ref{appendix: nfe}.
For sampling, the temperature is selected separately for each model on the validation set, whose pipeline is described in Appendix~\ref{appendix: temperature search}.

\subsection{Evaluations}
\label{sec: eval}

\textbf{Evaluation datasets.} 
We use the \textit{test-clean} subset of LibriTTS~\cite{libritts} as the validation set, following~\cite{sfm} to construct prompt-target pairs. 
For testing, we use the \textit{test-en} and \textit{test-zh} of Seed-TTS~\cite{Seed-TTS} test sets, and the \textit{en} and \textit{zh} subsets of CosyVoice~3~\cite{cosyvoice3} test sets.

\textbf{Objective evaluations.} 
We use UTMOS~\cite{utmos} to evaluate naturalness.
Following Seed-TTS~\cite{Seed-TTS}, we use Whisper-large-v3~\cite{whisper} to compute word error rate (WER) for English and Paraformer-zh~\cite{paraformer} to compute character error rate (CER) for Chinese.
For speaker similarity (SIM), we extract speaker embeddings using the WavLM-large~\cite{wavlm} speaker verification model and compute the cosine similarity.

\textbf{Subjective evaluation.} 
We conducted comparative MOS (CMOS, $[-3,3]$) and similarity MOS (SMOS, $[1,5]$) tests to evaluate naturalness and similarity, respectively. In each test, we recruited 20 native English or Chinese listeners on Prolific\footnote{\url{https://www.prolific.com}} with £9 per hour.

\textbf{Evaluation protocol.}
In our experiments, we observe that under a fixed framework, where the model backbone, training procedure, and inference pipeline are kept identical, it is uncommon for one model to outperform another across all objective metrics. Therefore, for hyperparameter selection based on objective evaluation, we consider one model better than another if it achieves better results in at least two of the three metric categories: UTMOS, WER/CER, and SIM.
\section{Experimental results and analysis}
\label{sec: results and analysis}

Table~\ref{tab: objective results} reports the objective results on the Seed-TTS and CosyVoice 3 test sets. 
\textit{Numerical KO} denotes our proposed numerical kinetic-optimal scheduler, \textit{Grid-searched} denotes the grid-searched heuristic MI-DFM scheduler selected on the validation set (Appendix~\ref{appendix: beta search}), and \textit{w/o corrector} indicates removal of the proposed finite-step moment correction.
Under the controlled setting with the same backbone, training procedure, and inference pipeline, GibbsTTS achieves the best overall performance among the evaluated masked discrete generative baselines.
On Seed-TTS test sets, GibbsTTS obtains the best results on all three metrics for test-en, and obtains the best UTMOS and SIM on test-zh with CER close to the best baseline. On CosyVoice 3 test sets, it achieves the best UTMOS for both languages and the best English WER, while remaining close to the best results on the other metrics. These results suggest that the metric-induced path is effective for codec-based zero-shot TTS when combined with the proposed kinetic-optimal scheduler and finite-step correction.

The results also show that scheduler choice is both important and decoding-dependent for masked generative baselines. For example, the MaskGCT scheduler $\kappa_t=\sin(\frac{\pi}{2}t)$ performs poorly with masked DFM, but gives the strongest overall performance among masked DD schedulers. Conversely, schedulers that work well for masked DFM do not necessarily work well for masked DD. This suggests that scheduler design is tightly coupled with the discrete generative formulation.

\textbf{Effect of the kinetic-optimal scheduler.}
Compared with the grid-searched MI-DFM scheduler, the numerical KO scheduler gives better overall performance across the two test sets. It improves all three metrics on Seed-TTS test-en, and improves UTMOS and SIM on test-zh while maintaining competitive CER. On CosyVoice 3, it brings clear gains over the searched scheduler in most metrics, with a small CER regression on the Chinese subset.
Importantly, the numerical KO scheduler does not require downstream hyperparameter search over hand-designed scheduler families. It therefore improves or matches empirical performance while removing a sensitive scheduler-selection procedure.

\textbf{Effect of the finite-step moment correction.}
The finite-step moment correction consistently improves MI-DFM across test sets and scheduler choices. Under the numerical KO scheduler, removing the corrector degrades all objective metrics on both Seed-TTS and CosyVoice 3, and the same trend holds for the grid-searched scheduler. The uniform direction and magnitude of the gap across all subset-metric combinations indicate that the gain is not specific to a particular language or scheduler.
These results are consistent with the motivation of the correction: a first-order CTMC solver is prone to finite-step path-tracking error, and adjusting the jump probability lets the sampler better follow the reference path while preserving the CTMC jump destination distribution.

\textbf{Subjective evaluation results.}
For the listening test, we focus on systems that isolate our two proposed methods, together with the strongest baselines of each family based on the objective evaluation results and the publicly released MaskGCT (\textit{MaskGCT (original)}) for reference.
The evaluated utterances are sampled in equal proportions from the Seed-TTS and CosyVoice 3 test sets.
Since CosyVoice 3 does not provide ground-truth target speech, we cannot include ground-truth samples in the subjective evaluations.
Table~\ref{tab: subjective results} reports the subjective evaluation results. Since GibbsTTS is used as the reference system in CMOS tests, its CMOS scores are zero. All compared systems obtain negative CMOS scores, indicating that listeners prefer GibbsTTS in naturalness. 
For speech similarity, GibbsTTS also achieves the highest SMOS on both English and Chinese evaluations. These subjective results support that the proposed scheduler and correction improve perceptual speech quality, rather than only automatic metrics.

\textbf{Additional analyses are provided in the appendices}, including grid search for the heuristic MI-DFM scheduler (Appendix~\ref{appendix: beta search}), shared versus
per-codebook scheduler construction (Appendix~\ref{appendix: share and per-codebook}), comparison with the per-layer
training strategy used in MaskGCT
(Appendix~\ref{appendix: training objective}).

\begin{table}[t]
  \caption{Objective evaluation results. The best value among generated systems is highlighted in \textbf{bold}.}
  \label{tab: objective results}
  \centering
  \vspace{-4pt}

  \begin{subtable}{\textwidth}
    \caption{Results on Seed-TTS test sets.}
    \label{tab: results seedtts}
    \vspace{-4pt}
    \centering
    \resizebox{1\textwidth}{!}{
    \setlength{\tabcolsep}{1.0mm}{
    \begin{tabular}{ll | ccc | ccc}
      \toprule
      \multirow{2}{*}{\textbf{Method}} 
      & \multirow{2}{*}{\textbf{Scheduler}} 
      & \multicolumn{3}{c}{\textbf{test-en}} 
      & \multicolumn{3}{c}{\textbf{test-zh}} \\
      \cmidrule(r){3-5} \cmidrule(r){6-8}
      && \textbf{UTMOS}$\uparrow$ 
      & \textbf{WER (\%) $\downarrow$} 
      & \textbf{SIM $\uparrow$} 
      & \textbf{UTMOS}$\uparrow$ 
      & \textbf{CER (\%) $\downarrow$} 
      & \textbf{SIM $\uparrow$} \\
      \midrule
      Ground truth & \textemdash
        & 3.527 & 2.020 & 0.734
        & 2.782 & 1.327 & 0.755 \\
      Codec reconstructed & \textemdash
        & 3.407 & 2.229 & 0.695
        & 2.564 & 1.472 & 0.725 \\
      \midrule
      MI-DFM (GibbsTTS) & Numerical KO
        & \textbf{3.651} & \textbf{1.777} & \textbf{0.743}
        & \textbf{2.712} & 1.327 & \textbf{0.790} \\
      MI-DFM w/o corrector & Numerical KO
        & 3.403 & 2.120 & 0.723
        & 2.447 & 1.777 & 0.775 \\
      MI-DFM & Grid-searched
        & 3.617 & 1.793 & 0.729
        & 2.628 & \textbf{1.297} & 0.784 \\
      MI-DFM w/o corrector & Grid-searched
        & 3.380 & 2.070 & 0.711
        & 2.381 & 1.637 & 0.767 \\
      \midrule
      Masked DFM & Closed-form KO
        & 3.639 & 1.969 & 0.742
        & 2.656 & 1.536 & 0.788 \\
      Masked DFM & DiFlow-TTS
        & 3.546 & 1.827 & 0.728 
        & 2.559 & 1.308 & 0.785 \\
      Masked DFM & MaskGCT
        & 3.269 & 2.724 & 0.712
        & 2.195 & 3.140 & 0.762 \\
      \midrule
      Masked DD & Closed-form KO
        & 3.634 & 5.808 & 0.731
        & 2.706 & 6.033 & 0.787 \\
      Masked DD & DiFlow-TTS
        & 2.768 & 9.303 & 0.672
        & 1.825 & 10.711 & 0.734 \\ 
      Masked DD & MaskGCT
        & 3.415 & 2.338 & 0.721
        & 2.387 & 1.583 & 0.776 \\
      \bottomrule
    \end{tabular}}}
  \end{subtable}

  \begin{subtable}{\textwidth}
    \caption{Results on CosyVoice 3 test sets. Ground truth targets are not provided in the sets.}
    \label{tab: results cosyvoice}
    \centering
    \vspace{-4pt}
    \resizebox{1.\textwidth}{!}{
    \setlength{\tabcolsep}{1.0mm}{
    \begin{tabular}{ll | ccc | ccc}
      \toprule
      \multirow{2}{*}{\textbf{Method}} 
      & \multirow{2}{*}{\textbf{Scheduler}} 
      & \multicolumn{3}{c}{\textbf{en}} 
      & \multicolumn{3}{c}{\textbf{zh}} \\
      \cmidrule(r){3-5} \cmidrule(r){6-8}
      && \textbf{UTMOS}$\uparrow$ 
      & \textbf{WER (\%) $\downarrow$} 
      & \textbf{SIM $\uparrow$} 
      & \textbf{UTMOS}$\uparrow$ 
      & \textbf{CER (\%) $\downarrow$} 
      & \textbf{SIM $\uparrow$} \\
      \midrule
      MI-DFM (GibbsTTS) & Numerical KO
        & \textbf{3.238} & \textbf{4.110} & 0.691
        & \textbf{2.438} & 4.144 & 0.780 \\
      MI-DFM w/o corrector & Numerical KO
        & 2.850 & 4.616 & 0.668
        & 2.135 & 5.485 & 0.772 \\
      MI-DFM & Grid-searched
        & 3.009 & 4.506 & 0.674
        & 2.189 & \textbf{3.706} & 0.772 \\
      MI-DFM w/o corrector & Grid-searched
        & 2.616 & 4.547 & 0.653
        & 1.939 & 4.274 & 0.755 \\
      \midrule
      Masked DFM & Closed-form KO
        & 3.049 & 5.162 & \textbf{0.695}
        & 2.294 & 4.855 & \textbf{0.781} \\
      Masked DFM & DiFlow-TTS
        & 2.925 & 4.288 & 0.673
        & 2.141 & 3.727 & 0.777 \\
      Masked DFM & MaskGCT
        & 2.354 & 8.767 & 0.614
        & 1.789 & 7.235 & 0.698 \\
      \midrule
      Masked DD & Closed-form KO
        & 3.042 & 18.353 & 0.677
        & 2.401 & 14.156 & 0.776 \\
      Masked DD & DiFlow-TTS
        & 1.885 & 36.133 & 0.562
        & 1.494 & 29.180 & 0.673 \\ 
      Masked DD & MaskGCT
        & 2.657 & 6.719 & 0.655
        & 1.903 & 4.575 & 0.762 \\
      \bottomrule
    \end{tabular}}}
  \end{subtable}
\end{table}

\begin{table}
  \caption{Subjective evaluation results. $^{\ddagger}$ indicates $p<0.05$, and $^{\dagger}$ indicates $p<0.1$ compared with GibbsTTS. The highest value for each metric is highlighted in \textbf{bold}.}
  \label{tab: subjective results}
  \centering
  \vspace{2pt}
\resizebox{0.7\textwidth}{!}{
\setlength{\tabcolsep}{1.0mm}{
\begin{tabular}{ll | ll | ll}
  \toprule
  \multirow{2}{*}{\textbf{Method}} 
  & \multirow{2}{*}{\textbf{Scheduler}} 
  & \multicolumn{2}{c}{\textbf{en}} 
  & \multicolumn{2}{c}{\textbf{zh}} \\
  \cmidrule(r){3-4} \cmidrule(r){5-6}
  && \textbf{CMOS}$\uparrow$ 
  & \textbf{SMOS}$\uparrow$ 
  & \textbf{CMOS}$\uparrow$ 
  & \textbf{SMOS}$\uparrow$ \\
  \midrule
MI-DFM (GibbsTTS) & Numerical KO
    & \textbf{0} (ref) & \textbf{4.18}
    & \textbf{0} (ref) & \textbf{4.28} \\
MI-DFM w/o corrector & Numerical KO
    & -0.362 $^{\ddagger}$ & 3.86 $^{\ddagger}$
    & -0.459 $^{\ddagger}$ & 4.03 $^{\dagger}$ \\
MI-DFM & Grid-searched
    & -0.257 $^{\ddagger}$ & 4.05
    & -0.153 & 4.21 \\
\midrule
Masked DFM & Closed-form KO
    & -0.229 $^{\dagger}$ & 3.94
    & -0.224 $^{\dagger}$ & 4.11 \\
Masked DFM & DiFlow-TTS
    & -0.286 $^{\ddagger}$ & 3.91 $^{\dagger}$
    & -0.247 $^{\dagger}$ & 4.10 \\
\midrule
Masked DD & MaskGCT
    & -0.743 $^{\ddagger}$ & 3.56 $^{\ddagger}$
    & -0.647 $^{\ddagger}$ & 3.87 $^{\ddagger}$\\
MaskGCT (original) & MaskGCT
    & -0.762 $^{\ddagger}$ & 4.06
    & -0.612 $^{\ddagger}$ & 4.19 \\
    \bottomrule
  \end{tabular}}}
\end{table}

\section{Comparison with SOTA zero-shot TTS systems}
\label{appendix: sota comparison}

\begin{table}[H]
  \caption{Objective evaluation results on Seed-TTS test sets. The best value among generated systems is highlighted in \textbf{bold}.}
\label{tab: sota seedtts}
\vspace{2pt}
\begin{flushleft}
\footnotesize
\textit{Note:} ``Emilia'' indicates whether Emilia is used for training. 
\textsc{Incl.} denotes that Emilia is included in the training data.
\end{flushleft}
  \centering
  \vspace{4pt}
\resizebox{1\textwidth}{!}{
\setlength{\tabcolsep}{1.0mm}{
  \begin{tabular}{lcc | ccc | ccc}
    \toprule
    \multirow{2}{*}{\textbf{Model}} 
  & \multirow{2}{*}{\textbf{Size}} 
  & \multirow{2}{*}{\textbf{Emilia}} 
  & \multicolumn{3}{c}{\textbf{test-en}} 
  & \multicolumn{3}{c}{\textbf{test-zh}} \\
  \cmidrule(r){4-6} \cmidrule(r){7-9}
  &&& \textbf{UTMOS}$\uparrow$ 
  & \textbf{WER (\%) $\downarrow$} 
  & \textbf{SIM $\uparrow$} 
  & \textbf{UTMOS}$\uparrow$ 
  & \textbf{CER (\%) $\downarrow$} 
  & \textbf{SIM $\uparrow$} \\
  \midrule
Ground truth & \textemdash & \textemdash
    & 3.527 & 2.020 & 0.734
    & 2.782 & 1.327 & 0.755 \\
 \midrule
\multicolumn{9}{l}{\textbf{\textit{NAR models with the Codec in MaskGCT}}} \\
Codec reconstructed & \textemdash & \textemdash
    & 3.407 & 2.229 & 0.695
    & 2.564 & 1.472 & 0.725 \\
\underline{GibbsTTS (ours)} & 0.4B & $\checkmark$
    & 3.651 & 1.777 & \textbf{0.743}
    & 2.712 & 1.327 & \textbf{0.790} \\
MaskGCT & 1.5B & $\checkmark$
    & 3.582 & 3.763 & 0.716
    & 2.647 & 2.217 & 0.774 \\ 
 \midrule
\multicolumn{9}{l}{\textbf{\textit{Other NAR models}}} \\
F5-TTS~\cite{f5-tts} & 0.3B & $\checkmark$
    & 3.762 & 2.020 & 0.654
    & 2.962 & 1.741 & 0.747 \\
OmniVoice~\cite{omnivoice} & 0.6B & $\times$
    & 3.896 & 1.475 & 0.741
    & 3.092 & \textbf{0.927} & 0.778 \\ 
\midrule
\multicolumn{9}{l}{\textbf{\textit{AR models}}} \\
Spark-TTS~\cite{spark-tts} & 0.5B & $\checkmark$
    & 3.947 & 2.221 & 0.572
    & 3.283 & 1.374 & 0.658 \\
FireRedTTS-2~\cite{fireredtts-2} & 1.5B & $\times$
    & 3.697 & 2.187 & 0.664
    & 2.842 & 1.145 & 0.728 \\
IndexTTS2~\cite{indextts2} & 1.5B & \textsc{Incl.}
    & 3.651 & 1.735 & 0.707 
    & 2.999 & 1.112 & 0.765 \\
CosyVoice 2~\cite{cosyvoice2} & 0.5B & $\times$
    & 4.164 & 2.506 & 0.656
    & 3.469 & 1.344 & 0.752 \\
CosyVoice 3~\cite{cosyvoice3} & 0.5B & $\times$
    & 3.949 & 2.145 & 0.696
    & 3.325 & 1.192 & 0.779 \\
Qwen3-TTS~\cite{qwen3-tts} & 0.6B & $\times$
    & 4.155 & 1.711 & 0.706
    & 3.477 & 1.129 & 0.765 \\
Qwen3-TTS & 1.7B & $\times$
    & \textbf{4.178} & \textbf{1.434} & 0.712
    & \textbf{3.505} & 1.100 & 0.770 \\
    \bottomrule
  \end{tabular}}}
\end{table}

\begin{table}[H]
  \caption{Objective evaluation results on CosyVoice 3 test sets. The best value for each metric is highlighted in \textbf{bold}.}
  \label{tab: sota cosyvoice}
  \centering
  \vspace{4pt}
\resizebox{1\textwidth}{!}{
\setlength{\tabcolsep}{1.0mm}{
  \begin{tabular}{lcc | ccc | ccc}
    \toprule
    \multirow{2}{*}{\textbf{Model}} 
  & \multirow{2}{*}{\textbf{Size}} 
  & \multirow{2}{*}{\textbf{Emilia}} 
  & \multicolumn{3}{c}{\textbf{en}} 
  & \multicolumn{3}{c}{\textbf{zh}} \\
  \cmidrule(r){4-6} \cmidrule(r){7-9}
  &&& \textbf{UTMOS}$\uparrow$ 
  & \textbf{WER (\%) $\downarrow$} 
  & \textbf{SIM $\uparrow$} 
  & \textbf{UTMOS}$\uparrow$ 
  & \textbf{CER (\%) $\downarrow$} 
  & \textbf{SIM $\uparrow$} \\
  \midrule
\multicolumn{9}{l}{\textbf{\textit{NAR models with the Codec in MaskGCT}}} \\
\underline{GibbsTTS (ours)} & 0.4B & $\checkmark$
    & 3.238 & 4.110 & 0.691
    & 2.438 & 4.144 & \textbf{0.780} \\ 
MaskGCT & 1.5B & $\checkmark$
    & 3.032 & 7.237 & 0.690
    & 2.357 & 6.647 & 0.773 \\ 
 \midrule
\multicolumn{9}{l}{\textbf{\textit{Other NAR models}}} \\
F5-TTS & 0.3B & $\checkmark$
    & 3.188 & 6.582 & 0.624
    & 2.337 & 5.867 & 0.741 \\
OmniVoice & 0.6B & $\times$
    & 3.619 & \textbf{3.182} & \textbf{0.705}
    & 2.884 & 2.735 & 0.766 \\
\midrule
\multicolumn{9}{l}{\textbf{\textit{AR models}}} \\
Spark-TTS & 0.5B & $\checkmark$
    & 3.587 & 6.840 & 0.499
    & 2.977 & 4.802 & 0.671 \\
IndexTTS2 & 1.5B & \textsc{Incl.}
    & 3.372 & 3.987 & 0.664
    & 2.544 & 3.064 & 0.762 \\
CosyVoice 2 & 0.5B & $\times$
    & 3.850 & 6.200 & 0.614
    & 3.042 & 3.939 & 0.748 \\
CosyVoice 3 & 0.5B & $\times$
    & 3.719 & 4.124 & 0.676
    & 2.899 & 3.166 & 0.776 \\
Qwen3-TTS & 0.6B & $\times$
    & 3.680 & 3.636 & 0.641
    & 3.221 & 2.900 & 0.755 \\
Qwen3-TTS & 1.7B & $\times$
    & \textbf{3.931} & 3.406 & 0.670
    & \textbf{3.328} & \textbf{2.653} & 0.754 \\
    \bottomrule
  \end{tabular}}}
\end{table}

Although the primary goal of our experiments is to validate the proposed
algorithms under controlled settings,
we further compare GibbsTTS with recent SOTA zero-shot TTS systems in 
Tables~\ref{tab: sota seedtts} and~\ref{tab: sota cosyvoice}. 
We use their official open-source checkpoints and released inference code.\footnote{FireRedTTS-2 results on the CosyVoice 3 
test sets are not reported. In our runs, its released inference code did not terminate on a subset of these utterances, and imposing a hard iteration cap substantially degraded generation quality.}

These results show that GibbsTTS is not the strongest system in terms of 
naturalness or WER/CER: larger-scale autoregressive (AR) models and recent NAR systems often achieve higher UTMOS or lower WER/CER. 
Nevertheless, GibbsTTS achieves the highest speaker similarity on three of the four evaluated test sets. On the remaining CosyVoice 3 English test set, it obtains the second-best similarity score and is only behind OmniVoice, a concurrent NAR zero-shot TTS system. This indicates that MI-DFM is particularly effective at preserving speaker identity in codec-based zero-shot TTS.

A more direct comparison can be made within the group of NAR models using the Codec in MaskGCT. Across both Seed-TTS and CosyVoice 3 test sets, GibbsTTS consistently outperforms MaskGCT on all reported objective metrics, despite using a smaller model size. The subjective evaluation in Table~\ref{tab: subjective results} shows the same tendency, with GibbsTTS being preferred in naturalness and obtaining higher similarity ratings than MaskGCT. Moreover, GibbsTTS uses a simpler token representation. MaskGCT involves both semantic tokens and codec tokens, whereas GibbsTTS performs generation directly over codec tokens only. 
These results suggest that the proposed MI-DFM-based model and overall
training-inference framework provide an effective way to model codec tokens
directly, without relying on an additional semantic-token stage.

The codec reconstruction results in Table~\ref{tab: sota seedtts} further indicate that codec choice itself affects the UTMOS range of codec-based TTS systems. Even before considering generative modeling errors, the reconstructed speech already shows a noticeable UTMOS gap from the ground-truth speech. When this work was initiated, MaskGCT was still the SOTA codec-based NAR zero-shot TTS system. We therefore conservatively adopt its codec rather than introducing a newer and stronger codec, so that our study focuses more on the comparison of discrete generative modeling strategies.

These comparisons should still be interpreted with caution. The external 
systems differ in model size, training data, tokenizer, codec, architecture, text frontend, and inference pipeline, so the results are not fully controlled. 
In particular, GibbsTTS uses a relatively simple G2P-based frontend, whereas some recent systems may benefit from more advanced text normalization. This can affect intelligibility-related metrics such as WER and CER, and therefore the SOTA comparison should not be interpreted as an isolated evaluation of the proposed discrete generative algorithms. The main evidence for the proposed algorithmic contributions comes from the controlled comparisons and ablation studies in the main experiments. The SOTA comparison mainly demonstrates the practical potential of GibbsTTS as a zero-shot TTS system.

\section{Limitations and future work}
\label{appendix: limitations}

Since the codec used in this work applies \(\ell_2\) normalization to token embeddings, the distance used in our metric-induced path is equivalent to a cosine distance. Other choices of token distance are not further investigated in this work. Furthermore, after introducing the kinetic-optimal scheduler, the construction of the distance matrix itself remains an important direction for future study. For example, we may further optimize the geometry of token embeddings and make it better suited to MI-DFM.

For an arbitrary discrete probability path, \cite{mi-dfm} provides the kinetic-optimal velocity formulation, and this work provides a numerical construction for the kinetic-optimal time scheduler. However, other types of probability paths remain unexplored.

In addition, as discussed in Section~\ref{sec: midfm correction}, we choose the local Fisher--Rao tangent statistic as the reference moment in our corrector. Other choices of reference moments may lead to better performance and deserve further investigation. Moreover, correctors are also discussed in \cite{mi-dfm}, providing additional insights into possible constructions. Therefore, corrector design, as well as more general decoding strategies for MI-DFM, remains a valuable direction for future work.

Finally, although the proposed algorithms are general, we only evaluate them on zero-shot TTS in this work. Their effectiveness in other domains remains to be explored.

\section{Conclusion}

We proposed two general algorithmic contributions for discrete flow matching: a kinetic-optimal time scheduler for prescribed scalar-parameterized discrete probability paths, and a finite-step moment correction for practical CTMC sampling that reduces path-tracking error while preserving the jump destination distribution. We instantiated both for metric-induced discrete flow matching and validated them on codec-based zero-shot TTS. Under controlled comparisons, the resulting GibbsTTS achieves the strongest overall performance against the evaluated masked discrete generative baselines.

\clearpage

\section{Acknowledgements}
This work was supported by JST Moonshot R\&D Grant Number JPMJMS2011 and JST SPRING Grant Number JPMJSP2108.

\bibliographystyle{abbrv}
\bibliography{reference.bib}

@article{f5-tts,
 author    = {Yushen Chen and Zhikang Niu and Ziyang Ma and Keqi Deng and Chunhui Wang and Jian Zhao and Kai Yu and Xie Chen},
 title     = {{F5-TTS:} {A} Fairytaler that Fakes Fluent and Faithful Speech with Flow Matching},
 year      = {2024},
 journal     = {arXiv preprint arXiv:2410.06885},
}

@article{cosyvoice2,
 author    = {Zhihao Du and Yuxuan Wang and Qian Chen and Xian Shi and Xiang Lv and Tianyu Zhao and Zhifu Gao and Yexin Yang and Changfeng Gao and Hui Wang and Fan Yu and Huadai Liu and Zhengyan Sheng and Yue Gu and Chong Deng and Wen Wang and Shiliang Zhang and Zhijie Yan and Jingren Zhou},
 title     = {{CosyVoice 2}: Scalable Streaming Speech Synthesis with Large Language Models},
 year      = {2024},
 journal     = {arXiv preprint arXiv:2412.10117},
}

@article{cosyvoice3,
 author    = {Zhihao Du and Changfeng Gao and Yuxuan Wang and Fan Yu and Tianyu Zhao and Hao Wang and Xiang Lv and Hui Wang and Chongjia Ni and Xian Shi and Keyu An and Guanrou Yang and Yabin Li and Yanni Chen and Zhifu Gao and Qian Chen and Yue Gu and Mengzhe Chen and Yafeng Chen and Shiliang Zhang and Wen Wang and Jieping Ye},
 title     = {{CosyVoice 3}: Towards In-the-wild Speech Generation via Scaling-up and Post-training},
 year      = {2025},
 journal     = {arXiv preprint arXiv:2505.17589},
}

@article{qwen3-tts,
 author    = {Hangrui Hu and Xinfa Zhu and Ting He and Dake Guo and Bin Zhang and Xiong Wang and Zhifang Guo and Ziyue Jiang and Hongkun Hao and Zishan Guo and Xinyu Zhang and Pei Zhang and Baosong Yang and Jin Xu and Jingren Zhou and Junyang Lin},
 title     = {{Qwen3-TTS} Technical Report},
 year      = {2026},
 journal     = {arXiv preprint arXiv:2601.15621},
}

@article{qwen3,
 author    = {An Yang and Anfeng Li and Baosong Yang and Beichen Zhang and Binyuan Hui and Bo Zheng and Bowen Yu and Chang Gao and Chengen Huang and Chenxu Lv and others},
 title     = {{Qwen3} Technical Report},
 year      = {2025},
 journal     = {arXiv preprint arXiv:2505.09388},
}

@inproceedings{utmos,
 author    = {Takaaki Saeki and Detai Xin and Wataru Nakata and Tomoki Koriyama and Shinnosuke Takamichi and Hiroshi Saruwatari},
 title     = {{UTMOS:} UTokyo-SaruLab System for VoiceMOS Challenge 2022},
 year      = {2022},
 pages     = {4521--4525},
 booktitle     = {Interspeech},
}

@inproceedings{whisper,
 author    = {Alec Radford and Jong Wook Kim and Tao Xu and Greg Brockman and Christine McLeavey and Ilya Sutskever},
 title     = {Robust Speech Recognition via Large-Scale Weak Supervision},
 year      = {2023},
 pages     = {28492--28518},
 booktitle     = {International Conference on Machine Learning (ICML)},
}

@inproceedings{vocos,
 author    = {Hubert Siuzdak},
 title     = {Vocos: Closing the gap between time-domain and Fourier-based neural vocoders for high-quality audio synthesis},
 year      = {2024},
 booktitle    = {International Conference on Learning Representations (ICLR)},
}

@inproceedings{adamw,
 author    = {Ilya Loshchilov and Frank Hutter},
 title     = {Decoupled Weight Decay Regularization},
 year      = {2019},
 booktitle    = {International Conference on Learning Representations (ICLR)},
}

@inproceedings{libritts,
  title={{LibriTTS}: A Corpus Derived from {LibriSpeech} for Text-to-Speech},
  author={Heiga Zen and Viet Dang and Rob Clark and Yu Zhang and Ron J. Weiss and Ye Jia and Zhifeng Chen and Yonghui Wu},
  booktitle={Proc. Interspeech},
  pages={1526--1530},
  year={2019},
}

@article{VALL-E,
  author       = {Chengyi Wang and
                  Sanyuan Chen and
                  Yu Wu and
                  Ziqiang Zhang and
                  Long Zhou and
                  Shujie Liu and
                  Zhuo Chen and
                  Yanqing Liu and
                  Huaming Wang and
                  Jinyu Li and
                  Lei He and
                  Sheng Zhao and
                  Furu Wei},
  title        = {Neural Codec Language Models are Zero-Shot Text to Speech Synthesizers},
  journal    = {arXiv preprint arXiv:2301.02111v1},
  year         = {2023},
}

@article{Seed-TTS,
  author       = {Philip Anastassiou and
                  Jiawei Chen and
                  Jitong Chen and
                  Yuanzhe Chen and
                  Zhuo Chen and
                  Ziyi Chen and
                  Jian Cong and
                  Lelai Deng and
                  Chuang Ding and
                  Lu Gao and
                  Mingqing Gong and
                  Peisong Huang and
                  Qingqing Huang and
                  Zhiying Huang and
                  Yuanyuan Huo and
                  Dongya Jia and
                  Chumin Li and
                  Feiya Li and
                  Hui Li and
                  Jiaxin Li and
                  Xiaoyang Li and
                  Xingxing Li and
                  Lin Liu and
                  Shouda Liu and
                  Sichao Liu and
                  Xudong Liu and
                  Yuchen Liu and
                  Zhengxi Liu and
                  Lu Lu and
                  Junjie Pan and
                  Xin Wang and
                  Yuping Wang and
                  Yuxuan Wang and
                  Zhen Wei and
                  Jian Wu and
                  Chao Yao and
                  Yifeng Yang and
                  Yuanhao Yi and
                  Junteng Zhang and
                  Qidi Zhang and
                  Shuo Zhang and
                  Wenjie Zhang and
                  Yang Zhang and
                  Zilin Zhao and
                  Dejian Zhong and
                  Xiaobin Zhuang},
  title        = {Seed-TTS: {A} Family of High-Quality Versatile Speech Generation Models},
  journal    = {arXiv preprint arXiv:2406.02430},
  year         = {2024},
}

@article{wavlm,
      title={{WavLM}: Large-Scale Self-Supervised Pre-Training for Full Stack Speech Processing}, 
      author={Sanyuan Chen and Chengyi Wang and Zhengyang Chen and Yu Wu and Shujie Liu and Zhuo Chen and Jinyu Li and Naoyuki Kanda and Takuya Yoshioka and Xiong Xiao and Jian Wu and Long Zhou and Shuo Ren and Yanmin Qian and Yao Qian and Jian Wu and Michael Zeng and Xiangzhan Yu and Furu Wei},
      volume       = {16},
      number       = {6},
      pages        = {1505--1518},
      year={2022},
      journal={IEEE Journal of Selected Topics in Signal Processing},
}

@inproceedings{dit,
 author    = {William Peebles and Saining Xie},
 title     = {Scalable Diffusion Models with Transformers},
 year      = {2023},
 booktitle     = {{IEEE/CVF} International Conference on Computer Vision (ICCV)},
}

@misc{stableTTS,
    author = {},
title={{StableTTS}},
year={2024},
  howpublished = {\url{https://github.com/KdaiP/StableTTS}}
}

@inproceedings{dfm,
 author    = {Itai Gat and Tal Remez and Neta Shaul and Felix Kreuk and Ricky T. Q. Chen and Gabriel Synnaeve and Yossi Adi and Yaron Lipman},
 title     = {Discrete Flow Matching},
 year      = {2024},
 booktitle     = {Annual Conference on Neural Information Processing Systems (NeurIPS)},
}

@inproceedings{mi-dfm,
 author    = {Neta Shaul and Itai Gat and Marton Havasi and Daniel Severo and Anuroop Sriram and Peter Holderrieth and Brian Karrer and Yaron Lipman and Ricky T. Q. Chen},
 title     = {Flow Matching with General Discrete Paths: {A} Kinetic-Optimal Perspective},
 year      = {2025},
 booktitle    = {International Conference on Learning Representations (ICLR)},
}

@inproceedings{generalized-masked-diffusion,
 author    = {Jiaxin Shi and Kehang Han and Zhe Wang and Arnaud Doucet and Michalis K. Titsias},
 title     = {Simplified and Generalized Masked Diffusion for Discrete Data},
 year      = {2024},
 booktitle     = {Annual Conference on Neural Information Processing Systems (NeurIPS)},
}

@inproceedings{var,
 author    = {Keyu Tian and Yi Jiang and Zehuan Yuan and Bingyue Peng and Liwei Wang},
 title     = {Visual Autoregressive Modeling: Scalable Image Generation via Next-Scale Prediction},
 year      = {2024},
 booktitle     = {Annual Conference on Neural Information Processing Systems (NeurIPS)},
}

@inproceedings{maskgct,
 author    = {Yuancheng Wang and Haoyue Zhan and Liwei Liu and Ruihong Zeng and Haotian Guo and Jiachen Zheng and Qiang Zhang and Xueyao Zhang and Shunsi Zhang and Zhizheng Wu},
 title     = {{MaskGCT}: Zero-Shot Text-to-Speech with Masked Generative Codec Transformer},
 year      = {2025},
 booktitle     = {International Conference on Learning Representations (ICLR)},
}

@inproceedings{emilia,
 author    = {Haorui He and Zengqiang Shang and Chaoren Wang and Xuyuan Li and Yicheng Gu and Hua Hua and Liwei Liu and Chen Yang and Jiaqi Li and Peiyang Shi and Yuancheng Wang and Kai Chen and Pengyuan Zhang and Zhizheng Wu},
 title     = {Emilia: An Extensive, Multilingual, and Diverse Speech Dataset for Large-Scale Speech Generation},
 year      = {2024},
 booktitle     = {{IEEE} Spoken Language Technology Workshop (SLT)},
}

@article{rope,
      title={{RoFormer}: Enhanced transformer with Rotary Position Embedding}, 
      author={Jianlin Su and Murtadha H. M. Ahmed and Yu Lu and Shengfeng Pan and Wen Bo and Yunfeng Liu},
      volume       = {568},
      pages        = {127063},
      year={2024},
      journal={Neurocomputing},
}

@article{swiglu,
	author={Noam Shazeer},
	title={{GLU} Variants Improve Transformer},
	year={2020},
	journal={arXiv preprint arXiv:2002.05202},
}

@inproceedings{rmsnorm,
 author    = {Biao Zhang and Rico Sennrich},
 title     = {Root Mean Square Layer Normalization},
 year      = {2019},
 booktitle     = {Annual Conference on Neural Information Processing Systems (NeurIPS)},
}

@inproceedings{dac,
 author    = {Rithesh Kumar and Prem Seetharaman and Alejandro Luebs and Ishaan Kumar and Kundan Kumar},
 title     = {High-Fidelity Audio Compression with Improved {RVQGAN}},
 year      = {2023},
 booktitle     = {Annual Conference on Neural Information Processing Systems (NeurIPS)},
}

@article{omnivoice,
 author    = {Han Zhu and Lingxuan Ye and Wei Kang and Zengwei Yao and Liyong Guo and Fangjun Kuang and Zhifeng Han and Weiji Zhuang and Long Lin and Daniel Povey},
 title     = {{OmniVoice}: Towards Omnilingual Zero-Shot Text-to-Speech with Diffusion Language Models},
 year      = {2026},
 journal     = {arXiv preprint arXiv:2604.00688},
}

@article{spark-tts,
 author    = {Xinsheng Wang and Mingqi Jiang and Ziyang Ma and Ziyu Zhang and Songxiang Liu and Linqin Li and Zheng Liang and Qixi Zheng and Rui Wang and Xiaoqin Feng and Weizhen Bian and Zhen Ye and Sitong Cheng and Ruibin Yuan and Zhixian Zhao and Xinfa Zhu and Jiahao Pan and Liumeng Xue and Pengcheng Zhu and Yunlin Chen and Zhifei Li and Xie Chen and Lei Xie and Yike Guo and Wei Xue},
 title     = {{Spark-TTS}: An Efficient LLM-Based Text-to-Speech Model with Single-Stream Decoupled Speech Tokens},
 year      = {2025},
 journal     = {arXiv preprint arXiv:2503.01710},
}

@article{fireredtts-2,
  author    = {Kun Xie and Feiyu Shen and Junjie Li and Fenglong Xie and Xu Tang and Yao Hu},
  title     = {{FireRedTTS-2}: Towards Long Conversational Speech Generation for Podcast and Chatbot},
  journal   = {arXiv preprint arXiv:2509.02020},
  year      = {2025},
}

@inproceedings{indextts2,
  author    = {Siyi Zhou and Yiquan Zhou and Yi He and Xun Zhou and Jinchao Wang and Wei Deng and Jingchen Shu},
  title     = {{IndexTTS2}: {A} Breakthrough in Emotionally Expressive and Duration-Controlled Auto-Regressive Zero-Shot Text-to-Speech},
  booktitle = {Proceedings of the AAAI Conference on
Artificial Intelligence},
  year      = {2026},
}

@article{encodec,
  author    = {Alexandre D{\'e}fossez and Jade Copet and Gabriel Synnaeve and Yossi Adi},
  title     = {High Fidelity Neural Audio Compression},
  journal   = {Transactions on Machine Learning Research},
  year      = {2022},
}

@book{fisher-amari,
  title={Information geometry and its applications},
  author={Shun-ichi Amari},
  volume={194},
  year={2016},
  publisher={Springer}
}

@inproceedings{paraformer,
 author    = {Zhifu Gao and ShiLiang Zhang and Ian McLoughlin and Zhijie Yan},
 title     = {Paraformer: Fast and accurate parallel transformer for non-autoregressive end-to-end speech recognition},
 year      = {2022},
 pages     = {2063–-2067},
 booktitle     = {Interspeech},
}

@inproceedings{sfm,
 author    = {Dong Yang and Yiyi Cai and Yuki Saito and Lixu Wang and Hiroshi Saruwatari},
 title     = {Shallow Flow Matching for Coarse-to-Fine Text-to-Speech Synthesis},
 year      = {2025},
 booktitle     = {Annual Conference on Neural Information Processing Systems (NeurIPS)},
}

@inproceedings{fudoki,
 author    = {Jin Wang and Yao Lai and Aoxue Li and Shifeng Zhang and Jiacheng Sun and Ning Kang and Chengyue Wu and Zhenguo Li and Ping Luo},
 title     = {{FUDOKI:} Discrete Flow-based Unified Understanding and Generation via Kinetic-Optimal Velocities},
 year      = {2025},
 booktitle     = {Annual Conference on Neural Information Processing Systems (NeurIPS)},
}

@article{ursa-metric-video,
 author    = {Haoge Deng and Ting Pan and Fan Zhang and Yang Liu and Zhuoyan Luo and Yufeng Cui and Wenxuan Wang and Chunhua Shen and Shiguang Shan and Zhaoxiang Zhang and Xinlong Wang},
 title     = {Uniform Discrete Diffusion with Metric Path for Video Generation},
 year      = {2025},
 journal     = {arXiv preprint arXiv:2510.24717},
}

@inproceedings{wam-flow,
 author    = {Yifang Xu and Jiahao Cui and Feipeng Cai and Zhihao Zhu and Hanlin Shang and Shan Luan and Mingwang Xu and Neng Zhang and Yaoyi Li and Jia Cai and Siyu Zhu},
 title     = {{WAM-Flow:} Parallel Coarse-to-Fine Motion Planning via Discrete Flow Matching for Autonomous Driving},
 year      = {2026},
 booktitle     = {{IEEE/CVF} International Conference on Computer Vision and Pattern Recognition (CVPR)},
}

@inproceedings{next-omni,
 author    = {Run Luo and Xiaobo Xia and Lu Wang and Longze Chen and Renke Shan and Jing Luo and Min Yang and Tat{-}Seng Chua},
 title     = {{NExT-OMNI:} Towards Any-to-Any Omnimodal Foundation Models with Discrete Flow Matching},
 year      = {2026},
 booktitle     = {International Conference on Learning Representations (ICLR)},
}

@article{diflow,
 author    = {Ngoc-Son Nguyen and Thanh V. T. Tran and Hieu-Nghia Huynh-Nguyen and Truong-Son Hy and Van Nguyen},
 title     = {{DiFlow-TTS}: Compact and Low-Latency Zero-Shot Text-to-Speech with Factorized Discrete Flow Matching},
 year      = {2024},
 journal     = {arXiv preprint arXiv:2407.05407},
}

@article{soundstream,
  author       = {Neil Zeghidour and Alejandro Luebs and Ahmed Omran and Jan Skoglund and Marco Tagliasacchi},
  title        = {{SoundStream}: An End-to-End Neural Audio Codec},
  journal      = {{IEEE}/{ACM} Transactions on Audio, Speech, and Language Processing},
  volume       = {30},
  pages        = {495--507},
  year         = {2022},
}

@article{soundstorm,
 author    = {Zal{\'{a}}n Borsos and Matthew Sharifi and Damien Vincent and Eugene Kharitonov and Neil Zeghidour and Marco Tagliasacchi},
 title     = {{SoundStorm}: Efficient Parallel Audio Generation},
 year      = {2023},
 journal     = {arXiv preprint arXiv:2305.09636},
}

@article{rvq,
  author       = {A. Vasuki and P.T. Vanathi},
  title        = {A review of vector quantization techniques},
  journal      = {{IEEE} Potentials},
  volume       = {25},
  pages        = {39--47},
  year         = {2006},
}

@article{scheduler-fisher,
 author    = {Leo Zhang and Saifuddin Syed},
 title     = {The Cosine Schedule is Fisher-Rao-Optimal for Masked Discrete Diffusion Models},
 year      = {2025},
 journal     = {arXiv preprint arXiv:2508.04884},
}

\appendix
\clearpage
\section{Related works}
\label{appendix: related works}

\textbf{Codec-based discrete acoustic modeling.}
Recent TTS systems that model speech with discrete acoustic representations mainly rely on neural audio codecs~\cite{soundstream,encodec,dac}, which convert waveforms into sequences of discrete codec tokens. 
Most modern neural codec models adopt residual vector quantization (RVQ)~\cite{rvq,soundstream}, where each frame is encoded by a stack of codebooks that progressively refine the quantization residual, so that earlier codebooks capture coarser acoustic information and later codebooks encode finer residual details.
VALL-E~\cite{VALL-E} formulated zero-shot TTS as language modeling over codec tokens, and subsequent autoregressive TTS systems~\cite{spark-tts,fireredtts-2,indextts2,qwen3-tts} further scaled this paradigm through larger-scale training and improved codec models.
Compared with AR codec-based TTS, NAR codec-based generation remains less explored.
Current NAR systems generate codec tokens mainly through masked generative modeling formulations~\cite{soundstorm,maskgct,diflow}.
Most of these works treat codec tokens as symbolic discrete
labels and rely on a learned embedding lookup, leaving the
geometric structure of the codec-token embedding space
unexploited at the generative modeling level. In this work, we
follow the codec-based TTS model route and build a NAR
TTS model on top of the acoustic codec released by MaskGCT,
whose encoder follows DAC~\cite{dac} and decoder
follows Vocos~\cite{vocos}. Different from prior
works, our method further exploits the geometry of the
codec-token latent space through metric-induced discrete
flow matching.

\textbf{Masked generative codec modeling.}
MaskGCT~\cite{maskgct} is a representative codec-based NAR zero-shot TTS model. It first predicts semantic tokens~\cite{wavlm} from the input text and then generates acoustic codec tokens conditioned on the predicted semantic tokens.
For discrete token modeling, MaskGCT adopts masked generative modeling and performs training and inference in a per-layer strategy over RVQ codebooks.
Our model differs from MaskGCT in three aspects.
First, we use a direct text-to-codec prediction paradigm without an intermediate semantic-token stage.
Second, instead of a mask-source probability path, our model generates codec tokens through metric-induced discrete flow matching, whose intermediate distributions exploit codec-token embedding geometry.
Third, we perform full-codebook training and inference, jointly predicting all RVQ codebooks at each generation step.
Accordingly, the MaskGCT-style masked discrete diffusion~\cite{generalized-masked-diffusion} and the standard masked discrete flow matching~\cite{dfm} are included as controlled baselines in our experiments.

\textbf{Concurrent work: OmniVoice.}
OmniVoice~\cite{omnivoice} is a concurrent zero-shot TTS model based on masked discrete diffusion language models. 
Similar to our work, it adopts a direct text-to-codec paradigm with a full-codebook NAR prediction strategy. 
The two works, however, have different focuses. 
OmniVoice targets omnilingual zero-shot TTS through system-level designs. In particular, instead of training a DiT-style backbone with explicit timestep conditioning from scratch, it initializes a bidirectional Transformer from a pre-trained LLM (Qwen3-0.6B~\cite{qwen3}).
This allows the model to inherit linguistic knowledge and subword tokenization from the pre-trained LLM, avoiding explicit grapheme-to-phoneme conversion and language-specific text normalization.
Combined with large-scale multilingual training, these designs lead to strong intelligibility on multilingual benchmarks. 
Our work instead focuses on the algorithmic aspects of discrete generative modeling.
Specifically, the proposed kinetic-optimal scheduler and finite-step moment correction for MI-DFM are orthogonal to system-level choices such as LLM initialization, tokenizer design, frontend simplification, multilingual scaling, and backbone architecture.
We therefore view OmniVoice as a complementary concurrent effort, and its engineering designs provide promising directions for future improvement of our system.





\clearpage
\section{Proofs}
\label{appendix: proof}

\setcounter{lemma}{0}
\counterwithin{lemma}{section}
\renewcommand{\thelemma}{\arabic{lemma}}
\setcounter{proposition}{0}

\begin{lemma}[Scheduler-invariant Fisher-Rao length]
\label{lemma: appendix length-invariance}
Let $p_t(x)=p(x;\kappa_t)$, where $p(x;\kappa)$ is a differentiable probability path with $\kappa\in[0,\kappa_{\max}]$, and $\kappa_t$ is a monotonic scheduler with $\kappa_0=0$ and $\kappa_1=\kappa_{\max}$. 
Then the Fisher-Rao length is invariant to the choice of $\kappa_t$:
\begin{equation}
    \ell[\kappa_t]
    =
    \int_0^1 \|\dot p_t\|_{g_{FR}}dt
    =
    \int_0^{\kappa_{\max}}\sqrt{\mathcal I(\kappa)}d\kappa
    \triangleq L,
\end{equation}
where
\begin{equation}
    \mathcal{I}(\kappa)
    =
    \sum_x
    \frac{
    \left(\partial_\kappa p(x;\kappa)\right)^2
    }{
    p(x;\kappa)
    }
\end{equation}
is the Fisher information of the path parameter $\kappa$.
\end{lemma}

\textit{Proof.}
By the chain rule,
\begin{equation}
    \dot{p}_t(x)
    =
    \partial_\kappa p(x;\kappa_t)\dot{\kappa}_t .
\end{equation}
Therefore, the Fisher-Rao speed satisfies
\begin{equation}
    \|\dot{p}_t\|_{g_{FR}}
    =
    \sqrt{\left(
    \sum_x
    \frac{
    \left(\partial_\kappa p(x;\kappa_t)\dot{\kappa}_t\right)^2
    }{
    p(x;\kappa_t)
    }
    \right)} 
    =
    |\dot{\kappa}_t|
    \sqrt{\mathcal{I}(\kappa_t)}.
\end{equation}
Since $\kappa_t$ is monotonic increasing, $\dot{\kappa}_t\geq 0$, and hence $|\dot{\kappa}_t|=\dot{\kappa}_t$.
Thus,
\begin{equation}
    \ell[\kappa_t]
    =
    \int_0^1
    \sqrt{\mathcal{I}(\kappa_t)}
    \dot{\kappa}_t dt.
\end{equation}
Using the change of variables $\kappa=\kappa_t$, we obtain
\begin{equation}
    \ell[\kappa_t]
    =
    \int_{\kappa_0}^{\kappa_1}
    \sqrt{\mathcal{I}(\kappa)}d\kappa
    =
    \int_0^{\kappa_{\max}}
    \sqrt{\mathcal{I}(\kappa)}d\kappa,
\end{equation}
which is independent of the specific scheduler
$\kappa_t$.
\qed

\vspace{30pt}

\begin{lemma}[Constant-speed characterization]
\label{lemma: appendix constant-speed}
For a prescribed probability path $p(x;\kappa)$, let
$p_t^*(x)=p(x;\kappa_t^*)$ be the path induced by the kinetic-optimal scheduler $\kappa_t^*$. Then $\kappa_t^*$ minimizes $\mathcal{E}_{FR}[\kappa_t]$ over $\kappa_t\in\mathcal{K}$ if and only if it satisfies
\begin{equation}
    \|\dot p_t^*\|_{g_{FR}} = L,
    \qquad \forall t\in[0,1],
\end{equation}
where $L$ is the scheduler-invariant Fisher-Rao length defined in
Lemma~\ref{lemma: appendix length-invariance}.
\end{lemma}

\textit{Proof.}
By the Cauchy-Schwarz inequality,
\begin{equation}
    \mathcal{E}_{FR}[\kappa_t]
    =
    \int_0^1 \|\dot p_t\|_{g_{FR}}^2 dt
    \geq
    \left(
    \int_0^1 \|\dot p_t\|_{g_{FR}} dt
    \right)^2.
\end{equation}
From Lemma~\ref{lemma: appendix length-invariance}, the length term is independent of
the scheduler and equals $L$. Hence,
\begin{equation}
    \mathcal{E}_{FR}[\kappa_t] \geq L^2.
\end{equation}
The equality holds if and only if $\|\dot p_t\|_{g_{FR}}$ is constant with
respect to $t$. Since
\begin{equation}
    \int_0^1 \|\dot p_t\|_{g_{FR}}dt = L,
\end{equation}
this constant must be $L$. 
Therefore, the kinetic-optimal scheduler satisfies
$\|\dot p_t\|_{g_{FR}}=L$.
\qed

\newpage
\begin{proposition}[Kinetic-optimal time scheduler]
\label{proposition: appendix ko scheduler}
Let $p_t(x)=p(x;\kappa_t)$ be a prescribed path parameterized by a monotonic scheduler $\kappa_t$, with $\kappa_0=0$ and $\kappa_1=\kappa_{\max}$.
Assume $\mathcal{I}(\kappa)>0$ on $[0,\kappa_{\max}]$. The kinetic-optimal scheduler and its time derivative are given by
\begin{equation}
    \kappa_t^*
    =
    F^{-1}(t),
    \quad 
    \dot{\kappa}_t^*
    =
    \frac{L}{
    \sqrt{\mathcal{I}(\kappa_t^*)}
    },
    \quad
    \text{where }
    F(\kappa)
    =
    \frac{
    \int_0^\kappa \sqrt{\mathcal{I}(\xi)} d\xi
    }{L}.
\end{equation}
\end{proposition}

\textit{Proof.}
By the chain rule,
\begin{equation}
    \dot p_t(x)
    =
    \partial_\kappa p(x;\kappa_t)\dot\kappa_t .
\end{equation}
Substituting this into the Fisher-Rao metric gives
\begin{equation}
    \|\dot p_t\|_{g_{FR}}^2
    =
    \dot\kappa_t^2
    \sum_x
    \frac{
    \left(\partial_\kappa p(x;\kappa_t)\right)^2
    }{
    p(x;\kappa_t)
    }
    =
    \dot\kappa_t^2 \mathcal{I}(\kappa_t).
\end{equation}
Since $\dot{\kappa}_t\geq 0$,
\begin{equation}
    \|\dot p_t\|_{g_{FR}}
    =
    \dot\kappa_t \sqrt{\mathcal{I}(\kappa_t)}.
\end{equation}
By Lemma~\ref{lemma: appendix constant-speed}, 
\begin{equation}
    \dot\kappa_t^* \sqrt{\mathcal{I}(\kappa_t^*)}
    =
    L,
    \quad
    \text{then }
    \dot\kappa_t^*
    =
    \frac{L}{\sqrt{\mathcal{I}(\kappa_t^*)}}.
\end{equation}
Integrating both sides and applying the change of variable yields
\begin{equation}
    \int_0^{\kappa_t^*}
    \sqrt{\mathcal{I}(\xi)}
    d\xi
    =
    Lt.
\end{equation}
Taking $t=1$ gives
\begin{equation}
    L
    =
    \int_0^{\kappa_{\max}}
    \sqrt{\mathcal{I}(\xi)}
    d\xi .
\end{equation}
Normalizing this accumulated length by the total length $L$, we define
\begin{equation}
    F(\kappa)
    =
    \frac{
    \int_0^{\kappa}
    \sqrt{\mathcal{I}(\xi)}
    d\xi
    }{
    \int_0^{\kappa_{\max}}
    \sqrt{\mathcal{I}(\xi)}
    d\xi
    }
    =
    \frac{
    \int_0^{\kappa}
    \sqrt{\mathcal{I}(\xi)}
    d\xi
    }{L}
    .
\end{equation}
Then, we obtain
\begin{equation}
    F(\kappa_t^*) = t.
\end{equation}
Since $\mathcal{I}(\kappa)>0$ on $[0,\kappa_{\max}]$, $F$ is strictly increasing and hence invertible, then
\begin{equation}
    \kappa_t^* = F^{-1}(t).
\end{equation}
\qed
\clearpage
\section{Kinetic-optimal scheduler construction for metric-induced paths}
\label{appendix: algo-numer}

Algorithm~\ref{algo: find_beta_max} determines $\beta_{\max}$ by requiring the endpoint distribution to be sufficiently concentrated on the target token. For each codebook $c$ and target token $x_1$, we evaluate the target probability $p_{\beta}^{(c)}(x_1\mid x_1)$ under the metric-induced conditional distribution, and define the endpoint concentration score as the minimum target probability over all codebooks and target tokens. Starting from an initial positive value, we double it until the score exceeds $1-\epsilon$, and then perform binary search on the resulting interval to find the smallest $\beta_{\max}$ satisfying the same criterion. This gives a finite numerical endpoint that approximates the limiting delta distribution up to tolerance $\epsilon$.

Algorithm~\ref{algo: KO_schedule} describes the numerical construction of the kinetic-optimal scheduler. Given the distance matrices $\{\mathbf{D}_c\}_{c=1}^{C}$, we first evaluate the metric-induced conditional distribution on a finite inverse-temperature grid $\{\beta_i\}_{i=1}^{I}$. For each grid point, the Fisher information is computed as the variance of the distance to the target token, and then averaged over target tokens and codebooks to obtain a global Fisher information $V_i$. Using these values, we approximate the cumulative Fisher--Rao arc length by numerical quadrature. The scheduler is then obtained by uniformly redistributing the accumulated arc length over $[0,1]$ and inverting the resulting $\ell$--$\beta$ relation by linear interpolation. Finally, the derivative table is computed from the constant-speed condition, yielding the lookup tables $\{\beta_j^*\}_{j=1}^{T}$ and $\{\dot{\beta}_j^*\}_{j=1}^{T}$.

\begin{algorithm}[H]
\caption{Determine the finite endpoint $\beta_{\max}$}
\label{algo: find_beta_max}
\begin{algorithmic}[1]
    \Require distance matrices $\{\mathbf{D}_c\in\mathbb{R}^{s\times s}\}_{c=1}^{C}$,
    endpoint tolerance $\epsilon=10^{-8}$,
    initial upper bound $\beta_{\mathrm{init}}=1$
    \Ensure Finite endpoint $\beta_{\max}$

    \State Define the metric-induced conditional distribution:
    \[
        p_{\beta}^{(c)}(x \mid x_1)
        =
        \frac{
        \exp\!\left(-\beta \mathbf{D}_c(x,x_1)\right)
        }{
        \sum_{y\in[s]}
        \exp\!\left(-\beta \mathbf{D}_c(y,x_1)\right)
        },
        \qquad c=1,\dots,C.
    \]

    \State Define the endpoint concentration score:
    \[
        f(\beta)
        =
        \min_{c\in\{1,\dots,C\}}
        \min_{x_1\in[s]}
        p_{\beta}^{(c)}(x_1 \mid x_1).
    \]

    \State $\beta_{\mathrm{hi}} \gets \beta_{\mathrm{init}}$

    \While{$f(\beta_{\mathrm{hi}}) < 1-\epsilon$}
        \State $\beta_{\mathrm{hi}} \gets 2\beta_{\mathrm{hi}}$
    \EndWhile

    \State Find the smallest $\beta_{\max}\in[0,\beta_{\mathrm{hi}}]$ by binary search such that
    \[
        f(\beta_{\max}) \geq 1-\epsilon.
    \]

    \State \Return $\beta_{\max}$
\end{algorithmic}
\end{algorithm}

\begin{algorithm}[H]
\caption{Compute the kinetic-optimal scheduler}
\label{algo: KO_schedule}
\begin{algorithmic}[1]
    \Require distance matrices $\{\mathbf{D}_c\}_{c=1}^{C}$,
    number of time points $T=1024$, grid size $I=4096$,
    maximum inverse temperature $\beta_{\max}$,
    numerical constant $\epsilon=10^{-8}$
    \Ensure scheduler table $\{\beta_j^*\}_{j=1}^{T}$,
    derivative table $\{\dot{\beta}_j^*\}_{j=1}^{T}$

    \State Construct a uniform inverse-temperature grid:
    $\quad
        0=\beta_1
        <
        \beta_2
        <
        \cdots
        <
        \beta_I
        =
        \beta_{\max}.
    $

    \For{$i=1,\dots,I$}
        \For{$c=1,\dots,C$}
            \For{$x_1\in[s]$}
                \State Compute the metric-induced conditional distribution
                \[
                    p_{\beta_i}^{(c)}(x\mid x_1)
                    =
                    \frac{
                    \exp\!\left(
                    -\beta_i \mathbf{D}_c(x,x_1)
                    \right)
                    }{
                    \sum_{y\in[s]}
                    \exp\!\left(
                    -\beta_i \mathbf{D}_c(y,x_1)
                    \right)
                    },
                    \quad x\in[s].
                \]

                \State Compute the conditional mean distance
                \[
                    \mu_{c,x_1}(\beta_i)
                    =
                    \sum_{x\in[s]}
                    p_{\beta_i}^{(c)}(x\mid x_1)
                    \mathbf{D}_c(x,x_1).
                \]

                \State Compute the conditional distance variance
                \[
                    \sigma_{c,x_1}^2(\beta_i)
                    =
                    \sum_{x\in[s]}
                    p_{\beta_i}^{(c)}(x\mid x_1)
                    \mathbf{D}_c(x,x_1)^2
                    -
                    \mu_{c,x_1}(\beta_i)^2.
                \]
            \EndFor

            \State Average over target tokens:
            \[
                V_{c,i}
                =
                \frac{1}{s}
                \sum_{x_1\in[s]}
                \sigma_{c,x_1}^2(\beta_i).
            \]
        \EndFor

        \State Average over codebooks to obtain the global Fisher information:
        \[
            V_i
            =
            \frac{1}{C}
            \sum_{c=1}^{C} V_{c,i}.
        \]
    \EndFor

    \State Compute the cumulative Fisher--Rao arc length by the trapezoidal rule:
    $\quad
        \ell_1 = 0.
    $

    \For{$i=2,\dots,I$}
        \State
        \[
            \ell_i
            =
            \ell_{i-1}
            +
            \frac{
            \sqrt{\max\{V_i,\epsilon\}}
            +
            \sqrt{\max\{V_{i-1},\epsilon\}}
            }{2}
            \left(
            \beta_i-\beta_{i-1}
            \right).
        \]
    \EndFor

    \State Construct uniform time points:
    $\quad
        t_j=\frac{j-1}{T-1},
        \quad j=1,\dots,T.
    $

    \For{$j=1,\dots,T$}
        \State Set the target Fisher--Rao arc length:
        $\quad
            \ell_j^* = t_j \ell_I.
        $

        \State Locate the interval:
        $\quad
            \ell_{i_j-1}
            \leq
            \ell_j^*
            \leq
            \ell_{i_j},
            \quad
            i_j\in\{2,\dots,I\}.
        $

        \State Linearly interpolate the inverse map between
        $(\ell_{i_j-1},\beta_{i_j-1})$ and
        $(\ell_{i_j},\beta_{i_j})$:
        \[
            a_j
            =
            \frac{
            \ell_j^*-\ell_{i_j-1}
            }{
            \ell_{i_j}-\ell_{i_j-1}+\epsilon
            },
            \quad
            \beta_j^*
            =
            \beta_{i_j-1}
            +
            a_j
            \left(
            \beta_{i_j}
            -
            \beta_{i_j-1}
            \right).
        \]

        \State Interpolate the Fisher information at $\beta_j^*$:
        \[
            \widetilde{a}_j
            =
            \frac{
            \beta_j^*-\beta_{i_j-1}
            }{
            \beta_{i_j}-\beta_{i_j-1}+\epsilon
            },
            \quad
            V_j^*
            =
            V_{i_j-1}
            +
            \widetilde{a}_j
            \left(
            V_{i_j}-V_{i_j-1}
            \right).
        \]

        \State Compute the derivative of the Fisher-uniform scheduler:
        \[
            \dot{\beta}_j^*
            =
            \frac{
            \ell_I
            }{
            \sqrt{\max\{V_j^*,\epsilon\}}
            }.
        \]
    \EndFor

    \State \Return $\{\beta_j^*\}_{j=1}^{T}$,
    $\{\dot{\beta}_j^*\}_{j=1}^{T}$

\end{algorithmic}
\end{algorithm}

\clearpage
\section{Recovering the closed-form kinetic-optimal scheduler for mixture paths}
\label{appendix: mixture KO}

This appendix applies Proposition~\ref{proposition: ko scheduler} to the mixture path and recovers its closed-form kinetic-optimal scheduler. This serves as a consistency check of the proposed formulation.

For the mixture $p_t(x)$ defined in Eq~\ref{eq: mixture pt}, we consider it as 
$p(x \mid x_1;\kappa)
=(1-\kappa)\,p(x) + \kappa\,\delta_{x_1}(x)$.

For a fixed target token $x_1$, let $p_1 = p(x_1)$, then
\begin{equation}
p(x_1 \mid x_1;\kappa) = p_1 + (1-p_1)\kappa,
\qquad
p(x \mid x_1;\kappa) = (1-\kappa)p(x), \ \ x \neq x_1,
\end{equation}
and the Fisher information with respect to $\kappa$ is
\begin{equation}
    \mathcal{I}_{x_1}(\kappa)
    =
    \sum_x
    \frac{\bigl(\partial p(x \mid x_1;\kappa)\bigr)^2}
         {p(x \mid x_1;\kappa)}
    =
    \frac{1-p_1}
    {(1-\kappa)\bigl(p_1 + (1-p_1)\kappa\bigr)},
    \label{eq: fisher_mix}
\end{equation}

and the cumulative Fisher-Rao arc-length $\ell(\kappa)$ is
\begin{equation}
    \ell(\kappa)
    = \int_0^\kappa \sqrt{\mathcal{I}_{x_1}(\xi)} d\xi 
    = \int_0^\kappa \sqrt{ \frac{1-p_1}{(1-\xi)\bigl(p_1 + (1-p_1)\xi\bigr)} } d\xi.
\end{equation}

Let $y(\xi) = p_1 + (1-p_1)\xi$, then 
\begin{align}
    \ell(\kappa) 
    &= \int_{p_1}^{y(\kappa)} \sqrt{ \frac{(1-p_1)^2}{y(1-y)} } \frac{dy}{1-p_1} \\
    &= \int_{p_1}^{y(\kappa)} \frac{1}{\sqrt{y(1-y)}} dy \\
    &= 2 \arcsin\left(\sqrt{y(\kappa)}\right) - 2 \arcsin(\sqrt{p_1}) \\
    &= 2 \arcsin\left(\sqrt{p_1 + (1-p_1)\kappa}\right) - 2 \arcsin(\sqrt{p_1}).
\end{align}

\begin{equation}
    L
    = \ell(1) 
    = 2 \left( \frac{\pi}{2} - \arcsin(\sqrt{p_1}) \right) 
    = 2 \arccos(\sqrt{p_1})
\end{equation}

\begin{equation}
    F(\kappa)
    = \frac{\arcsin\left(\sqrt{p_1 + (1-p_1)\kappa}\right) - \arcsin(\sqrt{p_1})}{\arccos(\sqrt{p_1})}
\end{equation}

By setting $F(\kappa_t) = t$ and solving for $\kappa_t$:
\begin{equation}
    \arcsin\left(\sqrt{p_1 + (1-p_1)\kappa_t}\right) 
    = t\arccos(\sqrt{p_1}) + \arcsin(\sqrt{p_1})
\end{equation}

Let $\Omega(x_1) = \arccos(\sqrt{p_1})$, then
\begin{equation}
    \arcsin\left(\sqrt{p_1 + (1-p_1)\kappa_t}\right) 
    = t\Omega(x_1) + \arcsin(\sqrt{p_1})
    = t\Omega(x_1) + \frac{\pi}{2} - \Omega(x_1)
\end{equation}

Then 
\begin{equation}
    \sqrt{p_1 + (1-p_1)\kappa_t}= \cos\left((1-t)\Omega(x_1)\right),
\end{equation}

\begin{equation}
    \kappa_t 
    = \frac{\cos^2\left((1-t)\Omega(x_1)\right) - p_1}{1 - p_1}
    = 1 + \frac{\cos^2\left((1-t)\Omega(x_1)\right) - 1}{1 - p_1}
    = 1 - \frac{\sin^2\left((1-t)\Omega(x_1)\right)}{\sin^2(\Omega(x_1))},
\end{equation}
which is the closed-form kinetic-optimal scheduler as shown in Eq.~\ref{eq: mixture ko}.
\clearpage
\section{Training and inference of GibbsTTS}
\label{appendix: algo-model}

\begin{algorithm}[H]
\caption{Training of GibbsTTS}
\label{algo: train}
\begin{algorithmic}[1]
    \Require model $\theta$, 
    number of codebooks $C$,
    precomputed scheduler table $\{\beta_j^*\}_{j=1}^{T}$,
    distance matrices $\{\mathbf{D}_c\}_{c=1}^{C}$

    \Repeat
        \State Sample batch $(x_1, \mathrm{cond})$
        \quad $\triangleright$ cond: text and language ID
        \State Sample $t\sim\mathcal{U}[0,1]$
        \State $\beta_t \leftarrow \textsc{LinearInterp}(t, \{\beta_j^*\}_{j=1}^{T})$

        \State Sample $r\sim\mathcal{U}(0,0.3)$
        \State $m \leftarrow \mathrm{round}(rN)$
        \quad $\triangleright$ $N$: valid token length
        \State Construct prediction mask
        \quad $\triangleright$ prompt region: 0, target region: 1
        \[
            M_i=\mathbbm{1}[i > m].
        \]

        \For{$c=1,\dots,C$ \textbf{in parallel}}
            \State Sample tokens
            \[
                \widetilde x_t^{i,c}
                \sim
                \mathrm{Categorical}
                \left(
                    \mathrm{softmax}_{x\in[s]}(-\beta_t
                \mathbf{D}_c(x, x_1^{i,c}))
                \right)
                \quad \triangleright \text{by Gumbel-Max trick}.
            \]
        \EndFor

        \State Construct input
        \[
            x_t^{i,c} 
            = 
            \widetilde x_t^{i,c} \text{ if } M_i = 1 \text{ else } x_1^{i,c}.
        \]

        \State Predict target-token distribution
        \[
            p_{1|t}^{\theta}
            =
            \theta(x_t,t,\mathrm{cond}).
        \]
        
        \State Define codebook-wise loss weights
        \[
            w_c
            =
            1-\frac{c-1}{C}.
        \]

        \State Compute loss
        \[
            \mathcal{L}(\theta)
            =
            \frac{
            \sum_{i=1}^{N}\sum_{c=1}^{C}
            M_i w_c
            \left[
            -\log p_{1|t}^{\theta,i,c}
            (
            x_1^{i,c}
            \mid x_t,t,\mathrm{cond}
            )
            \right]
            }{
            \sum_{i=1}^{N}\sum_{c=1}^{C} M_i w_c
            }.
        \]

        \State Update $\theta$ via gradient descent on $\mathcal{L}(\theta)$

    \Until{Training end}
\end{algorithmic}
\end{algorithm}

\begin{algorithm}[tb]
\caption{Inference of GibbsTTS}
\label{algo: infer}
\begin{algorithmic}[1]
    \Require model $\theta$, conditioning $\mathrm{cond}$,
    prompt $x_{\mathrm{pr}}$, target length $N$,
    number of steps $K$, sampling temperature $\tau$,
    distance matrices $\{\mathbf{D}_c\}_{c=1}^{C}$,
    scheduler tables $\{\beta_j^*\}_{j=1}^{T}$ and $\{\dot{\beta}_j^*\}_{j=1}^{T}$

    \State Construct time grid $t_k=k/K$, $k=0,\dots,K$.

    \State Initialize target tokens
    \[
        x_0^{i,c}\sim \mathcal{U}\{1,\dots,s\},
        \qquad i=1,\dots,N,\quad c=1,\dots,C .
    \]
    \State Set $x_t \leftarrow [x_{\mathrm{pr}},x_0]$

    \For{$k=0,\dots,K-1$}
        \State $t \leftarrow t_k$, \quad $h \leftarrow t_{k+1}-t_k$

        \State Interpolate
        \[
            \beta_t \leftarrow \textsc{LinearInterp}(t,\{\beta_j^*\}_{j=1}^{T}),
            \quad
            \beta_{t+h} \leftarrow \textsc{LinearInterp}(t+h,\{\beta_j^*\}_{j=1}^{T}),
        \]
        \[
            \dot{\beta}_t \leftarrow \textsc{LinearInterp}(t,\{\dot{\beta}_j^*\}_{j=1}^{T}).
        \]

        \For{$i=1,\dots,N$ and $c=1,\dots,C$ \textbf{in parallel}}
            \State Sample
            \[
                \hat{x}_1^{i,c}
                \sim
                p_{1|t}^{\theta,i,c, \tau}
                (\cdot\mid x_t,t,\mathrm{cond})
                \quad \triangleright \text{by Gumbel-Max trick}.
            \]

            \State Define
            \[  
                z=x_t^{i,c},
                \quad
                \hat{x}_1=\hat{x}_1^{i,c},
                \quad
                \text{and for }
                x\in[s],
                \text{ let }
                d_x = D_c(x, \hat{x}_1), 
                \quad d_z = D_c(z, \hat{x}_1).
            \]

            \State Compute
            \[
                p_t(x\mid \hat{x}_1)
                =
                \mathrm{softmax}_{x\in[s]}(-\beta_t d_x),
            \quad
                u_t(x,z\mid \hat{x}_1)
                =
                p_t(x\mid \hat{x}_1)
                \left[
                \dot{\beta}_t(d_z-d_x)
                \right]_+ .
            \]

            \State Compute jump intensity and destination distribution
            \[
                \lambda_t(z\mid \hat{x}_1)
                =
                \sum_{x\in [s]} u_t(x,z\mid \hat{x}_1),
                \qquad
                \pi_t(x\mid z,\hat{x}_1)
                =
                \frac{
                u_t(x,z\mid \hat{x}_1)
                }{
                \lambda_t(z\mid \hat{x}_1)
                } .
            \]

            \State Compute the first-order jump probability
            \[
                \rho_{\mathrm{base}}
                =
                1-\exp(-h\lambda_t(z\mid \hat{x}_1)).
            \]

            \State Compute the reference distribution
            \[
                p_{t+h}(x\mid \hat{x}_1)
                =
                \mathrm{softmax}_{x\in[s]}(-\beta_{t+h}d_x).
            \]

            \State Compute
            \[
                A
                =
                d_z-\sum_{x\in[s]}p_{t+h}(x\mid \hat{x}_1)d_x,
                \quad
                B
                =
                \sum_{x\in[s]}
                \pi_t(x\mid z,\hat{x}_1)(d_z-d_x),
                \quad
                \rho^\star=\frac{A}{B}.
            \]
            
            \State Set
            \[
                \rho
                =
                \begin{cases}
                \rho^\star,
                & \text{if } \lambda_t(z\mid\hat{x}_1)>0
                ,B \neq0,
                \text{ and } 0\leq \rho^\star\leq 1,\\
                \rho_{\mathrm{base}},
                & \text{otherwise}.
                \end{cases}
            \]

            \State Sample $Z_{\text{jump}}\sim\mathcal{U}[0,1]$.
            \If{$Z_{\text{jump}}\leq \rho$ and $\lambda_t(z\mid\hat{x}_1)>0$}
                \State Sample
                \[
                    x_{t+h}^{i,c}\sim \pi_t(\cdot\mid z,\hat{x}_1).
                \]
            \Else
                \State Keep the current token:
                \[
                    x_{t+h}^{i,c}\leftarrow z.
                \]
            \EndIf
        \EndFor
        \State Set $x_t \leftarrow x_{t+h}$.
    \EndFor
    \State \Return the target part of $x_t$
\end{algorithmic}
\end{algorithm}

\clearpage
\section{Exact recovery for the mixture path}
\label{appendix: mixture-consistency}

For the standard mixture path, both the instantaneous rate and the finite-step
transition probability admit closed forms. We show that, in this case, the
proposed moment correction recovers the exact finite-step transition over
\([t,t+h]\), and therefore introduces no additional approximation. This result
serves as a consistency check. The correction is mainly useful when the exact
finite-step transition is not available in a simple closed form, as in the
metric-induced paths considered in Section~\ref{sec: method2}.

We work with the standard mixture path
\[
p_t(x\mid x_0,x_1)
=
(1-\kappa_t)\delta_{x_0}(x)
+
\kappa_t\delta_{x_1}(x),
\]
where \(\kappa_t\) is monotone increasing and \(x_0\neq x_1\). In this case,
the CTMC has only one possible non-trivial transition, from \(x_0\) to \(x_1\),
with instantaneous rate
\[
    \lambda_t
    =
    \frac{\dot{\kappa}_t}{1-\kappa_t}.
\]

The exact probability of jumping from \(x_0\) to \(x_1\) over the finite interval
\([t,t+h]\), conditioned on \(X_t=x_0\), is
\[
\begin{aligned}
\rho_{\mathrm{exact}}
&=
1-\exp\!\left(
-\int_t^{t+h}
\frac{\dot{\kappa}_s}{1-\kappa_s}\,ds
\right)
=
\frac{\kappa_{t+h}-\kappa_t}{1-\kappa_t}.
\end{aligned}
\]
The standard first-order solver freezes the instantaneous rate at time \(t\),
giving
\[
\rho_{\mathrm{base}}
=
1-\exp(-h\lambda_t)
=
1-\exp\!\left(
-\frac{h\dot\kappa_t}{1-\kappa_t}
\right),
\]
which agrees with \(\rho_{\mathrm{exact}}\) only in the limit \(h\to0\).

We now show that the proposed moment correction recovers
\(\rho_{\mathrm{exact}}\). For the mixture path, the jump destination
distribution is deterministic:
\[
    \pi_t(\cdot\mid x_0,x_1)=\delta_{x_1}.
\]
We choose the scalar moment as the target-state indicator,
\[
    \phi_t(x\mid x_1)=\mathbbm{1}\{x=x_1\}.
\]
When the current state is \(x_0\), we have
\[
    \phi_t(x_0\mid x_1)=0,
    \qquad
    \bar\phi_t(x_0,x_1)
    =
    \mathbb E_{y\sim \pi_t(\cdot\mid x_0,x_1)}
    [\phi_t(y\mid x_1)]
    =
    1.
\]
We take the reference moment to be the state-conditional finite-step moment,
namely the probability of being at \(x_1\) at time \(t+h\) given \(X_t=x_0\):
\[
    m_{t+h}(x_0,x_1)
    =
    \mathbb P(X_{t+h}=x_1\mid X_t=x_0)
    =
    \frac{\kappa_{t+h}-\kappa_t}{1-\kappa_t}.
\]
Substituting these quantities into the generic moment-correction formula in
Eq.~\ref{eq: generic corr rho} gives
\[
    \rho^\star
    =
    \frac{
    m_{t+h}(x_0,x_1)-\phi_t(x_0\mid x_1)
    }{
    \bar\phi_t(x_0,x_1)
    -
    \phi_t(x_0\mid x_1)
    }
    =
    \frac{\kappa_{t+h}-\kappa_t}{1-\kappa_t}
    =
    \rho_{\mathrm{exact}}.
\]

Thus, for the standard mixture path, the correction recovers the known exact finite-step transition. Since this transition is already available in closed form, the correction is redundant in this special case and mainly serves as a consistency check.

For metric-induced paths, however, the corresponding state-conditional finite-step moment is generally not available in closed form, since it would require solving the finite-step transition kernel of a time-inhomogeneous CTMC. Section~\ref{sec: method2} therefore uses the unconditional expectation under \(p_{t+h}\) as a practical reference moment, which makes the metric-induced correction a finite-step approximation that tracks a chosen scalar moment, rather than an exact CTMC transition.

\clearpage
\section{Hyperparameter search of $\beta$}
\label{appendix: beta search}

As introduced in Section~\ref{sec: mi-dfm}, for the original MI-DFM method proposed in~\cite{mi-dfm}, we need to select well-performing $a$ and $c$ for \[\beta_t = c (\frac{t}{1 - t})^a.\]

To reduce training costs, we conduct hyperparameter search on the Base variants and select the hyperparameters according to their performance on the validation set.
Although we explored a wider range of candidate values, Table~\ref{tab: beta_search} reports only the values around the finally selected hyperparameter for clarity.

\begin{table}[H]
  \caption{Hyperparameter search results for the MI-DFM on the validation set. The best value for each metric is highlighted in \textbf{bold}.}
  \label{tab: beta_search}
  \centering
  \vspace{4pt}
\resizebox{0.5\textwidth}{!}{
\begin{tabular}{ll | ccc}
  \toprule
  \textbf{$a$}
  & \textbf{$c$}
  & \textbf{UTMOS}$\uparrow$ 
  & \textbf{WER (\%) $\downarrow$} 
  & \textbf{SIM $\uparrow$} \\
  \midrule
1 & 1
    & 3.688 & 3.190 & 0.685 \\
2.5 & 1
    & 3.973 & 3.139 & 0.704 \\
5 & 1
    & 3.995 & 3.200 & \textbf{0.720} \\
7.5 & 1
    & 3.859 & 3.397 & 0.705 \\
10 & 1
    & \textbf{4.002} & 3.442 & 0.718 \\
20 & 1
    & 3.882 & 4.345 & 0.699 \\
    \midrule
5 & 0.5
    & 3.522 & 24.594 & 0.670 \\
5 & 1.5
    & 3.903 & \textbf{3.074} & 0.687 \\
    \bottomrule
  \end{tabular}}
\end{table}

Following the evaluation protocol described in Section~\ref{sec: eval}, we select $a=5$ and $c=1$ for the Large variant experiments.

\section{Shared and per-codebook schedulers}
\label{appendix: share and per-codebook}

Since the proposed kinetic-optimal scheduler for MI-DFM does not require hyperparameter search, a natural question is whether each RVQ codebook should use an individual scheduler. We therefore compare a shared scheduler with per-codebook schedulers on the GibbsTTS-Base variant.

\begin{table}[H]
  \caption{Ablation study on shared and per-codebook kinetic-optimal schedulers. The best value for each metric is highlighted in \textbf{bold}.}
  \label{tab: scheduler_ablation}
  \centering
  \vspace{4pt}
\resizebox{1.\textwidth}{!}{
\begin{tabular}{l | ccc | ccc}
  \toprule
  \multirow{2}{*}{\textbf{Method}} 
  & \multicolumn{3}{c}{\textbf{Seed-TTS test-en}} 
  & \multicolumn{3}{c}{\textbf{CosyVoice 3 en}} \\
  \cmidrule(r){2-4} \cmidrule(r){5-7}
  & \textbf{UTMOS}$\uparrow$ 
  & \textbf{WER (\%) $\downarrow$} 
  & \textbf{SIM $\uparrow$} 
  & \textbf{UTMOS}$\uparrow$ 
  & \textbf{WER (\%) $\downarrow$} 
  & \textbf{SIM $\uparrow$} \\
  \midrule
Shared scheduler
    & 3.631 & 1.961 & \textbf{0.711}
    & \textbf{3.120} & \textbf{7.224} & \textbf{0.649} \\
Per-codebook scheduler
    & \textbf{3.657} & \textbf{1.902} & 0.705
    & 3.072 & 7.661 & 0.647 \\
    \bottomrule
  \end{tabular}}
\end{table}

As shown in Table~\ref{tab: scheduler_ablation}, the per-codebook scheduler does not bring consistent improvements over the shared scheduler. Although it slightly improves UTMOS and WER on Seed-TTS test-en, the shared scheduler performs better in speaker similarity and achieves better results on all metrics of CosyVoice 3 en.

Moreover, the per-codebook scheduler requires storing separate $\beta$ and $\dot{\beta}$ lookup tables for different RVQ codebooks, introducing additional memory overhead. Therefore, we use a shared kinetic-optimal scheduler in our final model.

\clearpage
\section{Training objective}
\label{appendix: training objective}

As discussed in Section~\ref{sec: model}, we adopt full-codebook training and inference, where all RVQ codebooks are predicted jointly. We further introduce codebook-wise loss weights in the training objective.

We first discuss the motivation for full-codebook training and inference. As mentioned in Appendix~\ref{appendix: related works}, the concurrent work OmniVoice~\cite{omnivoice} adopts a similar full-codebook prediction strategy. 
However, in our setting, where the model is trained from scratch with a relatively small number of parameters, treating all RVQ codebooks equally makes optimization unstable in the early stage of training. In particular, we observed gradient explosion when no codebook-wise weighting was used. 
Since earlier RVQ codebooks capture coarser acoustic information, while later codebooks mainly encode residual details with progressively reduced perceptual importance, we apply linearly decayed codebook-wise weights, as shown in Lines 12--13 of Algorithm~\ref{algo: train}. This reduces the contribution of later codebooks and stabilizes training, allowing the model to converge reliably.

We also compare our full-codebook strategy with the per-layer training and inference strategy used in MaskGCT. The results are shown in Table~\ref{tab: training_objective_ablation}.

\begin{table}[H]
  \caption{Ablation study on full-codebook and per-layer training/inference strategies. The best value for each metric is highlighted in \textbf{bold}.}
  \label{tab: training_objective_ablation}
  \centering
  \vspace{4pt}
\resizebox{1\textwidth}{!}{
\begin{tabular}{lr | ccc | ccc}
  \toprule
  \multirow{2}{*}{\textbf{Method}} 
  & \multirow{2}{*}{\textbf{NFE}} 
  & \multicolumn{3}{c}{\textbf{Seed-TTS test-en}} 
  & \multicolumn{3}{c}{\textbf{CosyVoice 3 en}} \\
  \cmidrule(r){3-5} \cmidrule(r){6-8}
  && \textbf{UTMOS}$\uparrow$ 
  & \textbf{WER (\%) $\downarrow$} 
  & \textbf{SIM $\uparrow$} 
  & \textbf{UTMOS}$\uparrow$ 
  & \textbf{WER (\%) $\downarrow$} 
  & \textbf{SIM $\uparrow$} \\
  \midrule
Full-codebook & 32
    & 3.631 & \textbf{1.961} & \textbf{0.711}
    & 3.120 & \textbf{7.224} & \textbf{0.649} \\
\midrule
\multirow{5}{*}{Per-layer}
& 32
    & 3.648 & 4.048 & 0.672
    & 3.048 & 14.366 & 0.593 \\
& 64
    & 3.745 & 3.218 & 0.674
    & \textbf{3.233} & 16.769 & 0.594 \\
& 66
    & \textbf{3.751} & 3.377 & 0.675
    & 3.226 & 16.988 & 0.592\\
& 82
    & 3.357 & 4.098 & 0.650
    & 2.727 & 15.677 & 0.572 \\
& 192
    & 3.536 & 4.056 & 0.663
    & 2.955 & 15.636 & 0.588 \\
    \bottomrule
  \end{tabular}}
\end{table}

Since the codec that we use contains 12 RVQ codebooks, the total NFE in the per-layer strategy is allocated across the codebooks. Following MaskGCT, we consider the allocation
\[
[40, 16, 1, 1, 1, 1, 1, 1, 1, 1, 1, 1],
\]
which results in 66 total steps. We also evaluate several other allocations:
\[
\begin{aligned}
32 \text{ steps}  &: [16, 6, 1, 1, 1, 1, 1, 1, 1, 1, 1, 1],\\
64 \text{ steps}  &: [40, 14, 1, 1, 1, 1, 1, 1, 1, 1, 1, 1],\\
82 \text{ steps}  &: [16, 6, 6, 6, 6, 6, 6, 6, 6, 6, 6, 6],\\
192 \text{ steps} &: [16, 16, 16, 16, 16, 16, 16, 16, 16, 16, 16, 16].
\end{aligned}
\]

The results show that, under our model and experimental setting, the per-layer strategy used in MaskGCT generally performs worse than our full-codebook strategy. Although per-layer inference can allocate more refinement steps to earlier codebooks, it is less efficient when the total NFE is fixed, since the steps are distributed across RVQ layers instead of jointly updating all codebooks at each step. In contrast, full-codebook inference updates all codebooks simultaneously, which leads to better intelligibility and speaker similarity in our experiments.

In addition, we initially attempted to include the autoregressive (AR) model as a baseline. However, even with codebook-wise loss weighting, the AR model failed to learn stable full-codebook prediction. The training loss did not decrease effectively, and the generated samples did not form normal speech. 
This observation suggests that, in our setting, NAR modeling is more suitable for full-codebook prediction across all RVQ codebooks.

\clearpage
\section{Duration predictor}
\label{appendix: duration predictor}

Although MaskGCT describes a flow-matching-based duration predictor in the paper, its released implementation uses a simple rule-based duration estimation:
\begin{equation}
    L(\text{target speech})
    =
    \frac{L(\text{prompt speech})}{L(\text{prompt phoneme})}
    L(\text{target phoneme}),
\end{equation}
where \(L(\cdot)\) denotes the sequence length. This rule assumes that the speech-to-phoneme length ratio is approximately consistent between the prompt and the target utterance.

We found that this heuristic works reasonably well in standard cases, but it can produce unreliable duration estimates when the prompt speaking rate is out of distribution. To improve robustness, we clip the prompt-derived ratio using the average ratio computed from the training set:
\begin{equation}
    r_{\mathrm{prompt}}
    =
    \frac{L(\text{prompt speech})}{L(\text{prompt phoneme})},
    \qquad
    r
    =
    \operatorname{clip}
    \left(
        r_{\mathrm{prompt}},
        \gamma \bar{r},
        \frac{\bar{r}}{\gamma}
    \right),
\end{equation}
\begin{equation}
    L(\text{target speech})
    =
    r L(\text{target phoneme}),
\end{equation}
where \(\bar{r}\) is the average speech-to-phoneme length ratio in the training set, and \(\gamma \in (0,1]\) controls the strength of the clipping. A smaller \(\gamma\) gives a wider valid range and therefore behaves closer to the original rule, while \(\gamma=1\) forces all utterances to use the global average ratio \(\bar{r}\).

In our training set, the average ratios \(\bar{r}\) are \(3.224\) and \(3.286\) for English and Chinese, respectively. We evaluate the effect of \(\gamma\) using GibbsTTS-Base on the validation set. To focus on challenging cases, we rank the utterances in the validation set by their prompt speech-to-phoneme ratio deviation and select the top \(10\%\) most out-of-distribution utterances. The results are shown in Table~\ref{tab:duration-clipping}.

\begin{table}[H]
  \caption{Effect of duration-ratio clipping on speaking-rate outliers. The utterances are derived from the validation set. The best value for each metric is highlighted in \textbf{bold}.}
  \label{tab:duration-clipping}
  \centering
  \vspace{4pt}
  \resizebox{0.5\textwidth}{!}{
  \begin{tabular}{l | ccc}
    \toprule
    \textbf{\(\gamma\)}
    & \textbf{UTMOS}\(\uparrow\) 
    & \textbf{WER (\%)}\(\downarrow\) 
    & \textbf{SIM}\(\uparrow\) \\
    \midrule
    No clipping & 3.992 & 5.015 & 0.686 \\
    0.2         & 3.992 & 5.015 & 0.686 \\
    0.4         & 3.992 & 5.015 & 0.686 \\
    0.6         & 3.998 & 4.523 & 0.688 \\
    0.8         & 4.017 & 4.130 & \textbf{0.698} \\
    1.0         & \textbf{4.048} & \textbf{3.835} & 0.686 \\
    \bottomrule
  \end{tabular}}
\end{table}

When \(\gamma \leq 0.4\), the results are identical to those without clipping, indicating that the clipping range is still too wide to affect the selected utterances. As \(\gamma\) increases, duration estimation becomes more constrained and the WER consistently improves. 
Although \(\gamma=1.0\) achieves the best UTMOS and WER in this subset, it makes the duration estimate entirely determined by the global average ratio \(\bar{r}\). This removes prompt-dependent speaking-rate variation and leads to an overly deterministic duration predictor, which may be undesirable for broader application scenarios. Therefore, we fix \(\gamma=0.8\) in all experiments as a practical trade-off between robustness and prompt-dependent duration modeling.

\clearpage
\section{Inference sampling temperature selection}
\label{appendix: temperature search}

During inference, the logits predicted by the model are divided by a sampling temperature before token sampling. We select this temperature on the validation set. In practice, we observe that the suitable temperature range differs between MI-DFM models and masked DFM or discrete diffusion models. Here, we use GibbsTTS-Large as an example to illustrate the temperature-search pipeline.

\begin{table}[H]
  \caption{Temperature search results for GibbsTTS-Large on the validation set. The best value for each metric is highlighted in \textbf{bold}.}
  \label{tab: temperature_search}
  \centering
  \vspace{4pt}
\resizebox{0.6\textwidth}{!}{
\begin{tabular}{c | ccc}
  \toprule
  \textbf{Temperature}
  & \textbf{UTMOS}$\uparrow$ 
  & \textbf{WER (\%) $\downarrow$} 
  & \textbf{SIM $\uparrow$} \\
  \midrule
0.2 & 4.039 & 3.215 & 0.749 \\
0.4 & 4.079 & \textbf{3.013} & 0.751 \\
0.6 & \textbf{4.084} & 3.033 & \textbf{0.752} \\
0.8 & 4.048 & 3.053 & 0.748 \\
1.0 & 3.860 & 3.250 & 0.737 \\
    \bottomrule
  \end{tabular}}
\end{table}

As shown in Table~\ref{tab: temperature_search}, temperature $0.6$ achieves the best UTMOS and SIM while maintaining a WER close to the best value. Following the evaluation protocol described in Section~\ref{sec: eval}, we use $0.6$ as the sampling temperature for GibbsTTS. For MI-DFM-based models, the selected temperature is consistently $0.6$ in our experiments. Masked generative baselines use lower temperatures: $0.1$ for masked DFM with the closed-form kinetic-optimal scheduler, and $0.2$ for the remaining models.

\clearpage
\section{Effect of number of function evaluations}
\label{appendix: nfe}

We further study the effect of the number of function evaluations (NFE) during inference. Besides the main 32 NFE setting, Tables~\ref{tab: nfe seedtts results} and~\ref{tab: nfe cosyvoice results} report additional results with 16 and 64 NFE.
For MI-DFM, the proposed moment corrector improves performance over the corresponding solver without the corrector in nearly all settings, particularly on UTMOS and WER/CER metrics.
The improvement is especially clear in intelligibility metrics at low and moderate NFE, where finite-step discretization errors are more significant. On UTMOS, the corrector provides consistent gains across all NFE settings.
With the numerical kinetic-optimal scheduler, MI-DFM with the corrector achieves the best UTMOS on both Seed-TTS and CosyVoice 3 test sets, showing that the proposed scheduler and corrector are effective for improving naturalness.

Compared with grid-searched heuristic schedulers, the numerical kinetic-optimal scheduler achieves competitive or better overall performance without downstream scheduler search. 
Although grid-searched schedulers can occasionally obtain slightly better intelligibility metrics, the numerical scheduler provides stronger naturalness and similarity in most settings while avoiding manual scheduler tuning.

Masked DFM also benefits from increasing NFE, but its performance depends strongly on the scheduler. 
In particular, the closed-form kinetic-optimal scheduler provides strong results for masked DFM, whereas MaskGCT-style schedules require larger NFE to approach competitive performance. 
Masked DD is more sensitive to the scheduler. Although the closed-form KO scheduler improves naturalness, it leads to much worse WER/CER, suggesting that high perceptual quality does not necessarily imply good intelligibility for this decoding formulation.

\begin{table}[h]
  \caption{Objective evaluation results on Seed-TTS test sets. The best value for each metric is highlighted in \textbf{bold}.}
  \centering
  \label{tab: nfe seedtts results}
  \vspace{4pt}
\resizebox{1.\textwidth}{!}{
\setlength{\tabcolsep}{1.0mm}{
\begin{tabular}{llc | ccc | ccc}
  \toprule
  \multirow{2}{*}{\textbf{Method}} 
  & \multirow{2}{*}{\textbf{Scheduler}} 
  & \multirow{2}{*}{\textbf{NFE}} 
  & \multicolumn{3}{c}{\textbf{test-en}} 
  & \multicolumn{3}{c}{\textbf{test-zh}} \\
  \cmidrule(r){4-6} \cmidrule(r){7-9}
  &&& \textbf{UTMOS}$\uparrow$ 
  & \textbf{WER (\%) $\downarrow$} 
  & \textbf{SIM $\uparrow$} 
  & \textbf{UTMOS}$\uparrow$ 
  & \textbf{CER (\%) $\downarrow$} 
  & \textbf{SIM $\uparrow$} \\
  \midrule
Ground truth & \textemdash & \textemdash
    & 3.527 & 2.020 & 0.734
    & 2.782 & 1.327 & 0.755 \\
Codec reconstructed & \textemdash & \textemdash
    & 3.407 & 2.229 & 0.695
    & 2.564 & 1.472 & 0.725 \\
\midrule
\multirow{3}{*}{MI-DFM (GibbsTTS)} & \multirow{3}{*}{Numerical KO}
& 16
    & 3.548 & 2.037 & 0.736
    & 2.594 & 1.569 & 0.785 \\
&& 32
    & 3.651 & 1.777 & \textbf{0.743}
    & 2.712 & 1.327 & 0.790 \\
&& 64
    & \textbf{3.704} & 1.785 & 0.742
    & \textbf{2.779} & 1.223 & 0.790 \\
\midrule
\multirow{3}{*}{MI-DFM w/o corrector} & \multirow{3}{*}{Numerical KO}
& 16
    & 3.288 & 2.363 & 0.715
    & 2.323 & 2.364 & 0.769 \\
&& 32
    & 3.403 & 2.120 & 0.723
    & 2.447 & 1.777 & 0.775 \\
&& 64
    & 3.450 & 2.037 & 0.727
    & 2.490 & 1.562 & 0.777 \\
\midrule
\multirow{3}{*}{MI-DFM} & \multirow{3}{*}{Grid-searched}
& 16
    & 3.474 & 1.869 & 0.721
    & 2.477 & 1.519 & 0.774 \\
&& 32
    & 3.617 & 1.793 & 0.729
    & 2.628 & 1.297 & 0.784 \\
&& 64
    & 3.683 & \textbf{1.760} & 0.730
    & 2.704 & 1.349 & 0.785 \\
\midrule
\multirow{3}{*}{MI-DFM w/o corrector} & \multirow{3}{*}{Grid-searched}
& 16
    & 3.245 & 2.112 & 0.706
    & 2.242 & 2.002 & 0.759 \\
&& 32
    & 3.380 & 2.070 & 0.711
    & 2.381 & 1.637 & 0.767 \\
&& 64
    & 3.421 & 1.928 & 0.717
    & 2.442 & 1.580 & 0.772 \\
\midrule
\multirow{3}{*}{Masked DFM} & \multirow{3}{*}{Closed-form KO}
& 16
    & 3.533 & 2.296 & 0.739
    & 2.519 & 2.650 & 0.782 \\
&& 32
    & 3.639 & 1.969 & 0.742
    & 2.656 & 1.536 & 0.788 \\
&& 64
    & 3.661 & 1.793 & \textbf{0.743}
    & 2.687 & 1.342 & \textbf{0.791} \\
\midrule
\multirow{3}{*}{Masked DFM} & \multirow{3}{*}{DiFlow-TTS}
& 16
    & 3.471 & 1.869 & 0.727
    & 2.453 & 1.507 & 0.781 \\
&& 32
    & 3.546 & 1.827 & 0.728 
    & 2.559 & 1.308 & 0.785 \\
&& 64
    & 3.588 & 1.810 & 0.728
    & 2.592 & \textbf{1.195} & 0.785 \\
\midrule
\multirow{3}{*}{Masked DFM} & \multirow{3}{*}{MaskGCT}
& 16
    & 2.800 & 5.272 & 0.673
    & 1.787 & 9.419 & 0.713 \\
&& 32
    & 3.269 & 2.724 & 0.712
    & 2.195 & 3.140 & 0.762 \\
&& 64
    & 3.449 & 2.405 & 0.724
    & 2.417 & 1.616 & 0.776 \\
\midrule
\multirow{3}{*}{Masked DD} & \multirow{3}{*}{Closed-form KO}
& 16
    & 3.469 & 9.127 & 0.720
    & 2.507 & 10.474 & 0.779 \\
&& 32
    & 3.634 & 5.808 & 0.731
    & 2.706 & 6.033 & 0.787 \\
&& 64
    & 3.682 & 4.802 & 0.732
    & 2.773 & 4.599 & 0.789 \\
\midrule
\multirow{3}{*}{Masked DD} & \multirow{3}{*}{DiFlow-TTS}
& 16
    & 2.191 & 28.671 & 0.587
    & 1.557 & 33.345 & 0.644 \\
&& 32
    & 2.768 & 9.303 & 0.672
    & 1.825 & 10.711 & 0.734 \\
&& 64
    & 3.052 & 4.861 & 0.697
    & 2.072 & 3.800 & 0.763 \\
\midrule
\multirow{3}{*}{Masked DD} & \multirow{3}{*}{MaskGCT}
& 16
    & 3.306 & 2.615 & 0.715
    & 2.273 & 1.999 & 0.774 \\
&& 32
    & 3.415 & 2.338 & 0.721
    & 2.387 & 1.583 & 0.776 \\
&& 64
    & 3.457 & 2.137 & 0.722
    & 2.442 & 1.502 & 0.778 \\
\bottomrule
  \end{tabular}}}
\end{table}

\begin{table}[h]
  \caption{Objective evaluation results on CosyVoice 3 test sets. The best value for each metric is highlighted in \textbf{bold}.}
  \centering
  \label{tab: nfe cosyvoice results}
  \vspace{4pt}
\resizebox{1.\textwidth}{!}{
\setlength{\tabcolsep}{1.0mm}{
\begin{tabular}{llc | ccc | ccc}
  \toprule
  \multirow{2}{*}{\textbf{Method}} 
  & \multirow{2}{*}{\textbf{Scheduler}} 
  & \multirow{2}{*}{\textbf{NFE}} 
  & \multicolumn{3}{c}{\textbf{en}} 
  & \multicolumn{3}{c}{\textbf{zh}} \\
  \cmidrule(r){4-6} \cmidrule(r){7-9}
  &&& \textbf{UTMOS}$\uparrow$ 
  & \textbf{WER (\%) $\downarrow$} 
  & \textbf{SIM $\uparrow$} 
  & \textbf{UTMOS}$\uparrow$ 
  & \textbf{CER (\%) $\downarrow$} 
  & \textbf{SIM $\uparrow$} \\
  \midrule
\multirow{3}{*}{MI-DFM (GibbsTTS)} & \multirow{3}{*}{Numerical KO}
& 16
    & 3.042 & 4.711 & 0.686
    & 2.304 & 5.546 & 0.781 \\
&& 32
    & 3.238 & 4.110 & 0.691
    & 2.438 & 4.144 & 0.780 \\
&& 64
    & \textbf{3.312} & 4.206 & 0.677
    & \textbf{2.486} & 4.069 & 0.768 \\
\midrule
\multirow{3}{*}{MI-DFM w/o corrector} & \multirow{3}{*}{Numerical KO}
& 16
    & 2.623 & 6.090 & 0.662
    & 1.961 & 7.345 & 0.765 \\
&& 32
    & 2.850 & 4.616 & 0.668
    & 2.135 & 5.485 & 0.772 \\
&& 64
    & 2.938 & 4.302 & 0.676
    & 2.173 & 5.047 & 0.777 \\
\midrule
\multirow{3}{*}{MI-DFM} & \multirow{3}{*}{Grid-searched}
& 16
    & 2.773 & 4.192 & 0.666
    & 2.036 & 4.165 & 0.767 \\
&& 32
    & 3.009 & 4.506 & 0.674
    & 2.189 & 3.706 & 0.772 \\
&& 64
    & 3.132 & 4.411 & 0.669
    & 2.279 & 3.768 & 0.765 \\
\midrule
\multirow{3}{*}{MI-DFM w/o corrector} & \multirow{3}{*}{Grid-searched}
& 16
    & 2.390 & 5.107 & 0.642
    & 1.806 & 5.266 & 0.744 \\
&& 32
    & 2.616 & 4.547 & 0.653
    & 1.939 & 4.274 & 0.755 \\
&& 64
    & 2.690 & 4.916 & 0.659
    & 1.982 & 4.007 & 0.759 \\
\midrule
\multirow{3}{*}{Masked DFM} & \multirow{3}{*}{Closed-form KO}
& 16
    & 2.819 & 6.923 & 0.681
    & 2.120 & 8.377 & 0.769 \\
&& 32
    & 3.049 & 5.162 & \textbf{0.695}
    & 2.294 & 4.855 & 0.781 \\
&& 64
    & 3.098 & 4.342 & 0.692
    & 2.332 & 4.247 & \textbf{0.782} \\
\midrule
\multirow{3}{*}{Masked DFM} & \multirow{3}{*}{DiFlow-TTS}
& 16
    & 2.783 & 4.302 & 0.665
    & 2.014 & 4.117 & 0.770 \\
&& 32
    & 2.925 & 4.288 & 0.673
    & 2.141 & 3.727 & 0.777 \\
&& 64
    & 2.970 & \textbf{3.960} & 0.670
    & 2.182 & \textbf{3.494} & 0.776 \\
\midrule
\multirow{3}{*}{Masked DFM} & \multirow{3}{*}{MaskGCT}
& 16
    & 1.911 & 21.330 & 0.533
    & 1.523 & 16.973 & 0.611 \\
&& 32
    & 2.354 & 8.767 & 0.614
    & 1.789 & 7.235 & 0.698 \\
&& 64
    & 2.658 & 5.216 & 0.653
    & 1.978 & 5.355 & 0.753 \\
\midrule
\multirow{3}{*}{Masked DD} & \multirow{3}{*}{Closed-form KO}
& 16
    & 2.730 & 26.751 & 0.662
    & 2.192 & 21.432 & 0.763 \\
&& 32
    & 3.042 & 18.353 & 0.677
    & 2.401 & 14.156 & 0.776 \\
&& 64
    & 3.144 & 14.939 & 0.681
    & 2.476 & 11.735 & 0.781 \\
\midrule
\multirow{3}{*}{Masked DD} & \multirow{3}{*}{DiFlow-TTS}
& 16
    & 1.535 & 100.451 & 0.432
    & 1.369 & 57.556 & 0.523 \\
&& 32
    & 1.885 & 36.133 & 0.562
    & 1.494 & 29.180 & 0.673 \\ 
&& 64
    & 2.201 & 13.328 & 0.624
    & 1.622 & 11.742 & 0.735 \\
\midrule
\multirow{3}{*}{Masked DD} & \multirow{3}{*}{MaskGCT}
& 16
    & 2.466 & 7.702 & 0.649
    & 1.812 & 6.052 & 0.755 \\
&& 32
    & 2.657 & 6.719 & 0.655
    & 1.903 & 4.575 & 0.762 \\
&& 64
    & 2.725 & 5.189 & 0.657
    & 1.964 & 4.623 & 0.766 \\
\bottomrule
  \end{tabular}}}
\end{table}



\end{document}